\documentclass[11pt,a4paper,onecolumn]{IEEEtran}
\usepackage{blindtext, graphicx}
\usepackage{tabularx}
\usepackage{pgf,tikz,pgfplots}
\usepackage[ruled,vlined]{algorithm2e}
\usepackage{hyperref}
\pgfplotsset{compat=1.14}
\usepackage{mathrsfs}
\usetikzlibrary{arrows}
\hypersetup{
 colorlinks,
 linkcolor={blue!100!black},
 citecolor={blue!100!black},
 urlcolor={blue!80!black}
}
\usetikzlibrary{arrows}
\usepackage{amsmath,amssymb,amsfonts,amsthm,graphicx,bm,bbm,dsfont}
\newtheorem{theorem}{Theorem}
\newtheorem{remark}{Remark}

\newtheorem{proposition}{Proposition}
\newtheorem{example}{Example}
\newtheorem{lemma}{Lemma}

\DeclareSymbolFont{bbold}{U}{bbold}{m}{n}
\DeclareSymbolFontAlphabet{\mathbbold}{bbold}
\newcommand{\1}{\mathbbold{1}}

\newcommand\blfootnote[1]{%
  \begingroup
  \renewcommand\thefootnote{}\footnote{#1}%
  \addtocounter{footnote}{-1}%
  \endgroup
}

\usepackage{color}

\newcommand{\cD}{\mathcal{D}}
\newcommand{\cE}{\mathcal{E}}

\newcommand{\boldc}{\mathbf{c}}
\newcommand{\boldd}{\mathbf{d}}
\newcommand{\bolde}{\mathbf{e}}

\newcommand{\boldm}{\mathbf{m}}

\newcommand{\bolds}{\mathbf{s}}

\newcommand{\boldv}{\mathbf{v}}

\newcommand{\boldx}{\mathbf{x}}
\newcommand{\boldy}{\mathbf{y}}

\begin{document}
\title{Two Deletion Correcting Codes\\from Indicator Vectors}

\author{{Jin Sima}, \IEEEauthorblockN{{Netanel Raviv}, and {Jehoshua Bruck}}}
	%\IEEEauthorblockA{
	%	Department of Electrical Engineering,
    %   California Institute of Technology, Pasadena 91125, CA, USA\\}}
\vspace{-0.13ex}

\maketitle

\begin{abstract}
Construction of capacity achieving deletion correcting codes has been a baffling challenge for decades. A recent breakthrough by Brakensiek \textit{et al}., alongside novel applications in DNA storage, have reignited the interest in this longstanding open problem. In spite of recent advances, the amount of redundancy in existing codes is still orders of magnitude away from being optimal. In this paper, a novel approach for constructing binary two-deletion correcting codes is proposed. By this approach, parity symbols are computed from indicator vectors (i.e., vectors that indicate the positions of certain patterns) of the encoded message, rather than from the message itself. Most interestingly, the parity symbols and the proof of correctness are a direct generalization of their counterparts in the Varshamov-Tenengolts construction. Our techniques require $7\log(n)+o(\log(n)$ redundant bits to encode an~$n$-bit message, which is near-optimal.
\end{abstract}

\blfootnote{The work was presented in part at the IEEE International Symposium on Information Theory, July 2018. 
The work was supported in part by NSF grant CCF-1717884.
The work of Netanel Raviv was supported in part by the postdoctoral fellowship of the Center for the Mathematics of Information (CMI), Caltech, and in part by the Lester-Deutsch postdoctoral fellowship.

Jin Sima is with the Electrical Engineering Department, California Institute of Technology, Pasadena, CA, 91125, Email: jsima@caltech.edu. 

Netanel Raviv is with the Electrical Engineering Department, California Institute of Technology, Pasadena, CA, 91125, Email: netanel.raviv@gmail.com. 

Jehoshua Bruck is with the Electrical Engineering Department, California Institute of Technology, Pasadena, CA, 91125, Email: bruck@caltech.edu. 
}
%\blfootnote{The work of Netanel Raviv was supported in part by the postdoctoral fellowship of the Center for the Mathematics of Information (CMI), Caltech, and in part by the Lester-Deutsch postdoctoral fellowship.}

%\IEEEpeerreviewmaketitle

\section{Introduction}
A \textit{deletion} in a binary sequence $\boldsymbol{c}=(c_1,\ldots,c_n)\in\{0,1\}^n$ is the case where a symbol is removed from~$\boldsymbol{c}$, which results in a subsequence length~$n-1$. Similarly, the result of a~$k$-deletion is a subsequence of~$\boldsymbol{c}$ of length~$n-k$.
%Denote $\boldsymbol{c}=(c_1,\ldots,c_n)\in\{0,1\}^n$ as a sequence of $n$ bits. A $k$-deletion error refers to the case when $k$ bits $c_{i_1},c_{i_2},\ldots,c_{i_k}$ in $\boldsymbol{c}$ turn into empty symbols. As a consequence of a k-deletion error, the sequence $\boldsymbol{c}$ of length $n$ becomes its subsequence of length $n-k$. 
A $k$-deletion code $\mathcal{C}$ is a set of~$n$-bit sequences, no two of which share a common subsequence of length $n-k$; and clearly, such a code can correct any~$k$-deletion.

It has been proved in \cite{levenshtein1966binary} that the largest size $L_k(n)$ of a~$k$-deletion code satisfies
\begin{equation}\label{eq1}
\frac{2^k(k!)^22^n}{n^{2k}}\lesssim L_k(n) \lesssim \frac{k!2^n}{n^k},
\end{equation}
which implies the existence of a~$k$-deletion code with at most $2k\log(n)+o(\log n)$ bits of redundancy for a constant~$k$.
%redundancy $O(k\log n)$ for constant~$k$ (\red{The constant should disappear when big-$O$ is applied. Shouldn't it be better to write~``with at most $2k\log(n)+o(\log n)$ bits of redundancy''?}).
However, to this day no explicit construction of such code is known beyond the case~$k=1$.

%To approach the lower bound in \eqref{eq1}, one can successively pick a sequence that does not share common subsequence of length $n-k$ with prior ones. This needs complexity exponential in $n$. Moreover, the construction of codes is not explicit.

For~$k=1$, the well-known Varshamov-Tenengolts (VT)~\cite{vt1965} construction
\begin{equation}\label{eq2}
\left\{\boldsymbol{c}: \sum^n_{i=1}ic_i= 0 \bmod (n+1)\right\}
\end{equation}
can correct one deletion with not more than $\log (n+1)$ bits of redundancy~\cite{levenshtein1966binary}. %There have been attempts to find explicit constructions of codes to correct beyond one deletion errors.
Several attempts to generalize the VT construction to~$k>1$ have been made. In the construction of~\cite{helberg2002multiple}, a modified Fibonacci sequence is used as weights instead of~$(1,2,\ldots,n)$ in~\eqref{eq2}. In~\cite{paluncic2012multiple}, number-theoretic arguments are used to obtain~$k$-deletion correction in run-length limited sequences. Yet, both~\cite{helberg2002multiple} and~\cite{paluncic2012multiple} have rates that are asymptotically bounded away from~$1$.

The problem of finding an explicit~$k$-deletion code of rate that approaches~$1$ as~$n$ grows has long been unsettled. Only recently, a code with $O(k^2\log k \log n)$ redundancy bits and encoding/decoding complexity\footnote{Here~$O_k$ denotes \textit{parameterized complexity}, i.e., $O_k(n\log^4 n)=f(k)O(n\log^4 n)$ for some function~$f$.} of $O_k(n\log^4 n)$
was proposed in~\cite{brakensiek2016efficient}. This code is based on a~$k$-deletion code of length~$\log n$, which is constructed using computer search. Nevertheless, the constants that are involved in the work of~\cite{brakensiek2016efficient} are orders of magnitude away from the lower bound in~\eqref{eq1} even for~$k=2$, and the code is not systematic. Moreover, finding a~$k$-deletion correcting code with an asymptotic rate 1 as an extension of the VT construction remains widely open\footnote{For~$k=2$,~\cite{Ryan} has very recently improved the redundancy up to~$8\log n$ using techniques similar to~\cite{brakensiek2016efficient}, our techniques incur lower redundancy and complexity, and use a fundamentally different approach.}.

%Finding a more explicit low redundancy time efficient $k$-deletion code construction, especially the extension of VT code to correct multiple errors, remains open even for $k=2$. 

One such potential extension is using higher order parity checks
%weights $i^e,e=1,2,\ldots,E$. For example, one can consider parity checks 
$\sum^n_{i=1}i^jc_i=0 \bmod( n^j+1)$ for $j=1,\ldots,t$, but counterexamples are easily constructible even for~$k=2$.
%for $E\le 4$ and it seems possible counterexamples exist for arbitrary $E$ even for two deletions \cite{brakensiek2016efficient}.
%\begin{example}
%\textup{Consider the two sequences $\boldsymbol{c}=[0~1~1~1~0~1~0~0~0~1~1~0]$ and $\boldsymbol{c}'=[1~0~0~1~1~1~0~1~0~0~0~1]$ of length $12$ that share a common subsequence $[0~1~1~1~0~1~0~0~0~1]$ of length $10$.
%Sequences $\boldsymbol{c}$ and $\boldsymbol{c}'$ have the same parity checks since $\sum^n_{i=1}ic_i = \sum^n_{i=1}ic'_i=26$, $\sum^n_{i=1}i^2c_i = \sum^n_{i=1}i^2c'_i=234$, and
%$\sum^n_{i=1}i^3c_i = \sum^n_{i=1}i^3c'_i=2366$.
%}
%\end{example}
In this paper, we find that similar higher order parity checks work when $t=3$, given that we restrict our attention to sequences with no consecutive ones. Consequently, applying these parity checks on certain \textit{indicator vectors} yields the desired result. For~$a$ and~$b$ in~$\{0,1\}$ and a binary sequence~$\boldsymbol{c}$, the~$ab$-indicator~$\1_{ab}(\boldsymbol{c})\in\{0,1\}^{n-1}$ of~$\boldsymbol{c}$ is
	\begin{align*}
		\1_{ab}(c)_i=\begin{cases}
			1 & \mbox{if }c_i=a \mbox{ and }c_{i+1}=b\\
			0 &\mbox{else}.
		\end{cases}.
	\end{align*} 
Since any two~$10$ or~$01$ patterns are at least two positions apart, the~$10$- and~$01$-indicators of any~$n$-bit sequence do not contain consecutive ones, and hence higher order parity checks can be applied.

The parity checks in the proposed code rely on the following integer vectors.
\begin{align*}
    \boldm^{(0)} &\triangleq (1,2,\ldots,n-1)\\
	\boldm^{(1)} &\triangleq \left(1,1+2,1+2+3,\ldots,\frac{n(n-1)}{2}\right)\\
	\boldm^{(2)} &\triangleq \left(1^2,1^2+2^2,1^2+2^2+3^2,\ldots,\frac{n(n-1)(2n-1)}{6}\right).
\end{align*}
Further, for~$\boldc\in\{ 0,1 \}^n$ let
\begin{alignat}{3}\label{equation:fgfunction}
 f(\boldc)&\triangleq (&&\1_{10}(\boldc) \cdot \boldm^{(0)} \bmod 2n,\nonumber\\
       &~           &&\1_{10}(\boldc) \cdot \boldm^{(1)} \bmod n^2,\nonumber\\
       &~           &&\1_{10}(\boldc) \cdot \boldm^{(2)} \bmod n^3),\mbox{ and}\nonumber\\
 h(\boldc)&\triangleq (&&\1_{01}(\boldc)\cdot \1 \bmod 3,\1_{01}(\boldc)\cdot \boldm^{(1)}\bmod 2n),
\end{alignat}
where~$\cdot$ denotes inner product over the integers, and~$\1$ denotes the all~$1$'s vector.

For any integer~$k$ let~$B_k(\boldc)$ be the \textit{k-deletion ball} of~$\boldc$, i.e., the set of~$n$-bit sequences that share a common~$n-k$ subsequence with~$\boldc$. The main result of this paper, from which a code construction is immediate, is as follows.

\begin{theorem}\label{theorem:maincoding}
    For~any integer~$n \ge 3$ %and$\boldc\in\{ 0,1 \}^n$, there exist a code length
and~$N=n+7\log n +o(\log n)$, there exists an encoding function~$\mathcal{E}:\{0,1\}^n\rightarrow \{0,1\}^{N}$ and a decoding function~$\mathcal{D}:\{0,1\}^{N-2}\rightarrow \{0,1\}^{n}$ such for any~$\boldc\in\{0,1\}^n$ and subsequence~$\boldc'\in \{0,1\}^{N-2}$ of~$\mathcal{E}(\boldc)$, we have~$\mathcal{D}(\boldc')=\boldc$. In addition, functions~$\mathcal{E}$ and~$\mathcal{D}$ can be computed in~$O(n)$ time. 
    %For~$n$-bit binary sequences~$\boldc$ and~$\boldc'$,
%    if ${\boldc\in B_2(\boldc')}$, $f(\boldc)=f(\boldc')$, and~$h(\boldc)=h(\boldc')$, then~$\boldc=\boldc'$.
    % \begin{align*}
    %     \boldc&\in B_2(\boldc')\\
    %     f(\boldc)&=f(\boldc'),\mbox{ and}\\
    %     h(\boldc)&=h(\boldc'),
    % \end{align*}
    % then~$\boldc=\boldc'$.
\end{theorem}
To prove this, we first show that the parities~$f(\boldc)$ and~$h(\boldc)$ can be used to correct two deletions.
\begin{theorem}\label{theorem:main}
    For~$\boldc,\boldc'\in\{ 0,1 \}^n$,
    %For~$n$-bit binary sequences~$\boldc$ and~$\boldc'$,
    if ${\boldc\in B_2(\boldc')}$, $f(\boldc)=f(\boldc')$, and~$h(\boldc)=h(\boldc')$, then~$\boldc=\boldc'$.
    % \begin{align*}
    %     \boldc&\in B_2(\boldc')\\
    %     f(\boldc)&=f(\boldc'),\mbox{ and}\\
    %     h(\boldc)&=h(\boldc'),
    % \end{align*}
    % then~$\boldc=\boldc'$.
\end{theorem}

Theorem~\ref{theorem:main} readily implies that that the functions~$h$ and~$f$ can serve as the redundancy bits in a~$2$-deletion code, and that the induced redundancy is at most $7\log(n)+o(\log n)$ (the additional term stems from protecting the redundancy bits). Furthermore, the encoding algorithm is trivial, and the decoding algorithm in Section~\ref{section:decoding} is linear. Most interestingly, the proof of Theorem~\ref{theorem:main} can be seen as a higher dimensional variant of the proof for the VT construction, as explained in the remainder of this section.

Clearly, a length~$n-1$ VT code can be seen as the set of sequences~$\boldc$ for which the values of~$\ell(\boldc)\triangleq \boldc \cdot \boldm^{(0)} \bmod n$ coincide. Adopting this point of view, the correctness of the VT construction can be proved by the following lemma, in which~$\ell_v(\boldc)\triangleq \boldc \cdot \boldv \bmod (v_{n-1}+1)$, and~$\boldv=(v_1,\ldots,v_{n-1})$ is a vector in~$\mathbb{Z}^{n-1}_+$. %The following lemma shows the correctness of a generalized VT code.
\begin{lemma}
For~$\boldc,\boldc'\in\{0,1\}^{n-1}$,~and~$\boldv\in \mathbb{Z}^{n-1}_+$, if~$\boldc\in B_1(\boldc')$,~$\ell_v(\boldc)=\ell_v(\boldc')$,~and~$v_1<v_2<\ldots<v_{n-1}$ then~$\boldc=\boldc'$.
\end{lemma}
In turn, the proof of this lemma can be completed by defining the following function. For a vector~$\boldv\in\mathbb{Z}^{n-1}_+$, an integer~$r\in[n-1]$, and a binary vector~$\boldx=(x_1,\ldots,x_s)$ with~$r+s-2\le n-1$, let
\begin{align}\label{equation:g}
    g_\boldv(r,\boldx)&\triangleq \boldx\cdot ((\boldv^{(r,r+s-2)},0)-(0,\boldv^{(r,r+s-2)}))\nonumber\\ &=x_1v_r-x_sv_{r+s-2}+\sum_{t=2}^{s-1}x_t(v_{t+r-1}-v_{t+r-2}),
\end{align}
where~$\boldv^{(r,r+s-2)}\triangleq (v_r,v_{r+1},\ldots,v_{r+s-2})$, and '$\cdot$' denotes inner product. Let~$k_1$ and~$k_2$ ($k_1<k_2$) be the indices of the deletions after which~$\boldc$ and~$\boldc'$ are identical.
Then we have
\begin{align}\label{equation:VTccprime}
	\boldc_{t} &= \1_{10}(\boldc')_t &\mbox{ if }& t<k_1\nonumber\\
	&~&\mbox{ or }&t>k_2,\mbox{ and}\nonumber\\
	\boldc_{t+1} &= \1_{10}(\boldc')_t &\mbox{ if }& k_1 \le t\le k_2-1.
\end{align}
One can find that
\begin{align}\label{equation:differenceofweightedsum}
\boldc \cdot \boldv - \boldc' \cdot \boldv = &\sum^{k_1-1}_{t=1}c_t v_t+c_{k_1}v_{k_1}+\sum^{k_2}_{t=k_1+1}c_t v_t+\sum^{n}_{t=k_2+1}c_t v_t\nonumber\\
&-\sum^{k_1-1}_{t=1}c_t v_t+\sum^{k_2-1}_{t=k_1}c_{t+1} v_t+c'_{k_2}v_{k_2}+\sum^{n}_{t=k_2+1}c_t v_t\nonumber\\
&=c_{k_1}v_{k_1} + \sum^{k_2}_{t=k_1+1}c_t(v_t-v_{t-1}) -c_{k_2}v_{k-2}\nonumber\\
&=g_{\boldv}(k_1,(\boldc^{(k_1,k_2)},c'_{k_2}))
\end{align}
%For~$\boldv=(1,\ldots,n)$, 
Hence, if~$\ell_v(\boldc)=\ell_v(\boldc')$ then~$g_{\boldv}(k_1,(\boldc^{(k_1,k_2)},c'_{k_2}))\equiv 0\bmod {(v_{n-1}+1)}$.
%c_{k_1},\ldots,c_{k_2-1}
Furthermore, since
\begin{align*}
-v_{n-1}&\le -v_{r+s-2}\le x_1v_r-x_sv_{r+s-2}+\sum_{t=2}^{s-1}x_t(v_{t+r-1}-v_{t+r-2}) \\
&\le v_r+\sum_{t=2}^{s-1}(v_{t+r-1}-v_{t+r-2})=v_{r+s-2}\le v_{n-1},
\end{align*}
it follows that~$\ell_v(\boldc)=\ell_v(\boldc')$ if and only if~$ g_{\boldv}(k_1,(\boldc^{(k_1,k_2)},c'_{k_2}))=0$. Therefore, the proof is concluded by the following lemma. 
\begin{lemma}\label{lemma:VT}
    For integers~$r$ and~$s$ such that~${r+s-2\le n-1}$, a vector $\boldv\in\mathbb{Z}^{n-1}_+$, and an~$s$-bit binary vector~$\boldx$,
    %, and an \emph{increasing} vector~$\boldv\in \bR^n$ (i.e.,~$v_{i+1}>v_{i}$ for all~$i\in[n-1]$), 
    if~$g_{\boldv}(r,\boldx)=0$ and~$v_1<\ldots<v_{n-1}$ then~$\boldx$ is a constant vector.
\end{lemma}
\begin{proof}
We dinstinguish between two cases according to the value of~$x_s$.
On the one hand,
if $x_s=0$, then it is readily verified that $g_{\boldv}(r,\boldx)$ is the sum of nonnegative terms. In which case, the equation $g_{\boldv}(r,\boldx)=0$ holds if and only if~$\boldx=0$.

On the other hand, if $x_s=1$, then
\begin{align}\label{Equation:lemmaVT}
g_{\boldv}(r,\boldx)&=v_{r}x_1+\sum^{s-1}_{t=2}(v_{t+r-1}-v_{t+r-2})x_t - v_{r+s-2}\nonumber\\
&\le v_{r}+\sum^{s-1}_{t=2}(v_{t+r-1}-v_{t+r-2}) - v_{r+s-2}=0.
\end{align}
The equality holds if and only if~$\boldx=1$. %$x_t=1,t=1,\ldots,s$, i.e., $\boldx$ is a $1$ vector.
\end{proof}
\begin{remark}
The VT code is the special case when $\boldv=(1,\ldots,n)$. From Lemma~\ref{lemma:VT} we have that~$c_{k_1}=\ldots=c_{k_2}=c'_{k_2}$. According to Equation~\ref{equation:VTccprime}, this implies that~$c'_t=c_{t+1}=c_t$ for~$k_1\le t\le k_2-1$ and~$c'_t=c_t$ for~$t<k_1$ or~$t\ge k_2$.
\end{remark}

The crux of proving Theorem~\ref{theorem:main} boils down to the following higher dimensional variant of Lemma~\ref{lemma:VT}.

\begin{lemma}\label{lemma:g}
    For integers~$r_1,r_2,s_1$, and~$s_2$ such that~$r_2>r_1+s_1-2$ and~$r_2+s_2-2\le n-1$, and binary sequences~$\boldx$ and~$\boldy$ of lengths~$s_1$ and~$s_2$, respectively, if
     \begin{align}\label{equation:lemmagEq}
        g_{\boldm^{(0)}}(r_1,\boldx)+\lambda g_{\boldm^{(0)}}(r_2,\boldy)&=0\mbox{, and}\nonumber\\
        g_{\boldm^{(1)}}(r_1,\boldx)+\lambda g_{\boldm^{(1)}}(r_2,\boldy)&=0,
    \end{align}
    where~$\lambda=\pm 1$, then~$\boldx$ and~$\boldy$ are constant vectors.
\end{lemma}

Additional technical claims, which involve the remaining ingredients of the redundancy bits, are given in the sequel. 

\section{Outline}
The proof of Theorem~\ref{theorem:main} is separated to the following two lemmas. In a nutshell, it is shown that for two confusable sequences, i.e., that share a common~$n-2$ subsequence, if the~$f$ redundancies coincide, then so are the~$10$-indicators. Then, it is shown that confusable sequences with identical~$10$-indicators and identical~$h$ redundancy have identical~$01$-indicators.

\begin{lemma}\label{lemma:10indicator}
For~$\boldc$ and~$\boldc'$ in~$\{0,1\}^n$, if ~$\boldc\in B_2(\boldc')$ and~$f(\boldc)=f(\boldc')$, then~$\1_{10}(\boldc)=\1_{10}(\boldc')$.
\end{lemma}

\begin{lemma}\label{lemma:01indicator}
For~$\boldc$ and~$\boldc'$ in~$\{0,1\}^n$ such that~$\boldc\in B_2(\boldc')$, if~$\1_{10}(\boldc)=\1_{10}(\boldc')$ and $h(\boldc)=h(\boldc')$, then $\1_{01}(\boldc)=\1_{01}(\boldc')$.
\end{lemma}

From these lemmas it is clear that two~$n$-bit sequences that share a common~$n-2$ subsequence and agree on the redundancies~$f$ and~$h$ have identical~$10$- and~$01$-indicators, and hence the next simple lemma concludes the proof of Theorem~\ref{theorem:main}.

\begin{lemma}\label{lemma:indicator}
For~$\boldc$ and~$\boldc'$ in~$\{0,1\}^n$ such that~$\boldc\in B_2(\boldc')$, if~$\1_{10}(\boldc)=\1_{10}(\boldc')$ and~ $\1_{01}(\boldc)=\1_{01}(\boldc')$ then~$\boldc=\boldc'$.
\end{lemma}
\begin{proof}
The conditions~$\mathbbm{1}_{10}(\boldsymbol{c})=\mathbbm{1}_{10}(\boldsymbol{c}')$ and~$\mathbbm{1}_{01}(\boldsymbol{c})=\mathbbm{1}_{01}(\boldsymbol{c}')$ imply that the ascending and descending transition positions of~$\mathbbm{1}_{01}(\boldsymbol{c})$ coincide with those of~$\mathbbm{1}_{01}(\boldsymbol{c}')$ respectively. Hence if transitions happen in~$\boldsymbol{c}$ or~$\boldsymbol{c}'$, then~$\boldsymbol{c}=\boldsymbol{c}'$. If no transitions happen in~$\boldsymbol{c}$ or~$\boldsymbol{c}'$ and~$\boldsymbol{c}\ne\boldsymbol{c}'$, then one of~$\boldsymbol{c}$ and~$\boldsymbol{c}'$ is all $0$'s vector and the other is all $1$'s vector. Since all $0$'s vector does not share a common subsequence of length~$n-2$ with all~$1$'s vector, we conclude that~$\boldsymbol{c}=\boldsymbol{c}'$.
\end{proof}
The proofs of Lemma~\ref{lemma:10indicator} and Lemma~\ref{lemma:01indicator} make extensive use of the following two technical claims, that are easy to prove.

\begin{lemma}\label{lemma:twoDeletions10}
For~$\boldc$ and~$\boldc'$ in~$\{0,1\}^n$, if~$\boldc\in B_2(\boldc')$ then $\1_{10}(\boldc)\in B_2(\1_{10}(\boldc'))$ and~$\1_{01}(\boldc)\in B_2(\1_{01}(\boldc'))$.
\end{lemma}
\begin{proof}
We first show that if~$\boldsymbol{c}\in B_1(\boldsymbol{c}')$ then~$\mathbbm{1}_{10}(\boldsymbol{c})\in B_1(\mathbbm{1}_{10}(\boldsymbol{c}'))$ and~$\mathbbm{1}_{01}(\boldsymbol{c})\in B_1(\mathbbm{1}_{01}(\boldsymbol{c}'))$. To this end, it suffices to show that if~$\boldd\in\{0,1\}^{n-1}$ is obtained from~$\boldsymbol{c}$ by one deletion, then~$\mathbbm{1}_{10}(\boldd)$ ($\mathbbm{1}_{01}(\boldd)$) is obtained from~$\mathbbm{1}_{10}(\boldsymbol{c})$ ($\mathbbm{1}_{01}(\boldsymbol{c})$) by one deletion (see table~\ref{table:indicatordeletion}).
	\begin{table}\begin{center}
		\begin{tabular}{|c|c|c|c|c|c|c|c|c|} \hline
            $c_{i-1}c_{i}c_{i+1}$ & $0\mathbf{0}0$ & $0\mathbf{0}1$ & $0\mathbf{1}0$ & $0\mathbf{1}1$ & $1\mathbf{0}0$ & $1\mathbf{0}1$ & $1\mathbf{1}0$ & $1\mathbf{1}1$ \\\hline	
            $\mathbbm{1}_{10}(c)_{i-1}\mathbbm{1}_{10}(c)_{i}$& $0\mathbf{0}$ & $0\mathbf{0}$ & $0\mathbf{1}$ & $0\mathbf{0}$ & $1\mathbf{0}$ & $\mathbf{1}0$ & $\mathbf{0}1$ & $0\mathbf{0}$\\\hline
            $\mathbbm{1}_{01}(c)_{i-1}\mathbbm{1}_{01}(c)_{i}$& $0\mathbf{0}$ & $0\mathbf{1}$ & $\mathbf{1}0$ & $\mathbf{1}0$ & $0\mathbf{0}$ & $0\mathbf{1}$ & $\mathbf{0}0$ & $0\mathbf{0}$\\\hline
           \end{tabular}
           \end{center}
           \caption{All possible cases of deletions of~$c_i$ for~$2\le i \le n-1$ correspond to deletions in~$\mathbbm{1}_{10}(\boldsymbol{c})$. The deleted symbol is in bold.}\label{table:indicatordeletion}
           \end{table}

Further, it is easy to see that a deletion of~$c_1$ corresponds to a deletion of~$\mathbbm{1}_{10}(c)_1$ (resp. $\mathbbm{1}_{01}(c)_1$) and a deletion of~$c_n$ corresponds to a deletion of~$\mathbbm{1}_{10}(c)_{n-1}$ (resp. $\mathbbm{1}_{01}(c)_{n-1}$). Hence, it follows that if
	\begin{align*}
		\boldc &\overset{1 \mbox{ \small{del'}}}{\longrightarrow}\boldd\overset{1 \mbox{ \small{del'}}}{\longrightarrow}\bolde\\
		\boldc'&\overset{1 \mbox{ \small{del'}}}{\longrightarrow}\boldd'\overset{1 \mbox{ \small{del'}}}{\longrightarrow}\bolde
	\end{align*}
	then
	\begin{align*}
		\mathbbm{1}_{10}(\boldc) &\overset{1 \mbox{ \small{del'}}}{\longrightarrow}\mathbbm{1}_{10}(\boldd)\overset{1 \mbox{ \small{del'}}}{\longrightarrow}\mathbbm{1}_{10}(\bolde)\\
		\mathbbm{1}_{10}(\boldc')&\overset{1 \mbox{ \small{del'}}}{\longrightarrow}\mathbbm{1}_{10}(\boldd')\overset{1 \mbox{ \small{del'}}}{\longrightarrow}\mathbbm{1}_{10}(\bolde)\\
		\mathbbm{1}_{01}(\boldc) &\overset{1 \mbox{ \small{del'}}}{\longrightarrow}\mathbbm{1}_{01}(\boldd)\overset{1 \mbox{ \small{del'}}}{\longrightarrow}\mathbbm{1}_{01}(\bolde)\\
		\mathbbm{1}_{01}(\boldc')&\overset{1 \mbox{ \small{del'}}}{\longrightarrow}\mathbbm{1}_{01}(\boldd')\overset{1 \mbox{ \small{del'}}}{\longrightarrow}\mathbbm{1}_{01}(\bolde),
	\end{align*}
	which concludes the claim.
\end{proof}
\begin{lemma}\label{lemma:equivMod3}
For~$\boldc,\boldc'\in \{ 0,1 \}^n$, if~$\boldc\in B_2(\boldc')$ and~$\1_{01}(\boldc)\cdot \1 = \1_{01}(\boldc')\cdot\1 \bmod 3$, then~$\1_{01}(\boldc)\cdot \1=\1_{01}(\boldc')\cdot \1$.% Here~$ab=10$ or~$01$.
\end{lemma}
\begin{proof}
Since~$\boldsymbol{c}\in B_2(\boldsymbol{c}')$ it follows from Lemma~\ref{lemma:twoDeletions10} that~$\mathbbm{1}_{10}(\boldsymbol{c})\in B_2 (\mathbbm{1}_{10}(\boldsymbol{c}'))$, and thus~$\mathbbm{1}_{10}(\boldsymbol{c})$ and~$\mathbbm{1}_{10}(\boldsymbol{c}')$ have a mutual~$(n-3)$-bit string~$\bolds$. Clearly,
	\begin{align*}
		\bolds\cdot \mathbbm{1} &\le \mathbbm{1}_{10}(\boldsymbol{c})\cdot \mathbbm{1}\le \bolds\cdot \mathbbm{1} +2\mbox{, and}\\
		\bolds\cdot \mathbbm{1} &\le \mathbbm{1}_{10}(\boldsymbol{c}')\cdot \mathbbm{1}\le \bolds\cdot \mathbbm{1} +2,
	\end{align*}
	and thus~$|\mathbbm{1}_{10}(\boldsymbol{c})\cdot \mathbbm{1}-\mathbbm{1}_{10}(\boldsymbol{c}')\cdot \mathbbm{1}|\le 2$. However, since~$3$ divides~$|\mathbbm{1}_{10}(\boldsymbol{c})\cdot \mathbbm{1}-\mathbbm{1}_{10}(\boldsymbol{c}')\cdot \mathbbm{1}|$, we must have that $\mathbbm{1}_{10}(\boldsymbol{c})\cdot \mathbbm{1}=\mathbbm{1}_{10}(\boldsymbol{c}')\cdot \mathbbm{1}$.
\end{proof}
In addition, one of the cases of the proof of Lemma~\ref{lemma:10indicator} requires a specialized variant of Lemma~\ref{lemma:g}.
\begin{lemma}\label{lemma:g+}
Let $r_1,r_2,s_1,s_2$ and $s_3$ be positive integers that satisfy $r_2=r_1+s_1$ and $r_2+s_2+s_3\le n-1$, and let $\boldx\in\{0,1\}^{s_1+s_2+1}$ and~$\boldy\in\{0,1\}^{1+s_2+s_3}$ be such that
\begin{align*}
    (x_{s_1+1},x_{s_1+2},\ldots,x_{s_1+s_2})=(y_{2},y_{3},\ldots,y_{s_2+1}),
\end{align*}
and $(x_{s_1+1},x_{s_1+2},\ldots,x_{s_1+s_2})$ has no adjacent $1$'s. If
\begin{alignat}{2}\label{equation:lemmag+}%\label{lemma3}
&g_{\boldm^{(0)}}(r_1,\boldx) + g_{\boldm^{(0)}}(r_2,\boldy) &&= 0,\nonumber\\
&g_{\boldm^{(1)}}(r_1,\boldx) + g_{\boldm^{(1)}}(r_2,\boldy) &&= 0,\mbox{ and}\nonumber\\
&g_{\boldm^{(2)}}(r_1,\boldx) + g_{\boldm^{(2)}}(r_2,\boldy) &&= 0,
\end{alignat}
then either~$x_1=\ldots=x_{s_1+s_2+1}=y_{1}=\ldots =y_{s_2+s_3+1}$ or
\begin{align}\label{equation:casecXY}
&x_1=x_2=\ldots=x_{s_1+1}=1-y_1,\nonumber\\
&x_{t}+x_{t+1}=1,\mbox{ for }t\in \{s_1+1,\ldots,s_1+s_2-1\}, \nonumber\\
&x_{s_1+s_2+1}+y_{s_2+1}=1,\mbox{ and}\nonumber\\
&y_{s_2+1}=\ldots = y_{s_2+s_3+1}.
\end{align}
\end{lemma}
The following lemma shows a property of~$g_{\boldv}(r,\boldx)$, which will be useful in the proof of Lemma~\ref{lemma:g} and Lemma~\ref{lemma:g+} that are given in Section~\ref{section:gProofs}.
\begin{lemma}\label{lemma:gproperty}
For integers~$r$ and~$s$ such that~${r+s-2\le n-1}$, a vector $\boldv$, and an~$s$-bit binary vector~$\boldx$,
    %, and an \emph{increasing} vector~$\boldv\in \bR^n$ (i.e.,~$v_{i+1}>v_{i}$ for all~$i\in[n-1]$),
    if~$g_{\boldv}(r,\boldx)=0$, then~$g_{\boldv}(r,\overline{\boldx})=0$, where~$\overline{\boldx}\triangleq\1-\boldx$.
\end{lemma}
\begin{proof}
Since
\begin{align}
g_{\boldv}(r,\boldx)&=v_{r}x_1+\sum^{s-1}_{t=2}(v_{t+r-1}-v_{t+r-2})x_t - v_{r+s-2}x_s\nonumber\\
&=v_{r}x_1+\sum^{s-1}_{t=2}(v_{t+r-1}-v_{t+r-2})x_t - v_{r+s-2}x_s - v_r -\sum^{s-1}_{t=2}(v_{t+r-1}-v_{t+r-2}) +v_{r+s-2}\nonumber\\
&=v_{r}(x_1-1)+\sum^{s-1}_{t=2}(v_{t+r-1}-v_{t+r-2})(x_t-1) - v_{r+s-2}(x_s-1)=-g_{\boldv}(r,\overline{\boldx})\label{equation:opposite}
\end{align}
Hence if~$g_{\boldv}(r,\boldx)=0$, we have~$g_{\boldv}(r,\overline{\boldx})=0$.
\end{proof}
%Finally, the following lemma about the determinant of a matrix with sum of columns shall be used to prove Lemma~\ref{lemma:g+}.
%\begin{lemma}\label{lemma:determinant}
% Let~$\bolda_1,\bolda_2,\ldots,\bolda_m\in\mathbb{R}^{n\times 1}$ be~$m$ column vectors and~$S_1,S_2,\ldots,S_n\subset\{1,2,\ldots,n\}$ be~$n$ sets. Then we have
% \begin{align*}
% \det((\sum_{i_1\in S_1}\bolda_{i_1},\sum_{i_2\in S_2}\bolda_{i_2} ,\ldots,\sum_{i_n\in S_n}\bolda_{i_n}))=\sum_{i_1\in S_1,i_2\in S_2,\ldots, i_n\in S_n}\det((\bolda_{i_1},\bolda_{i_2},\ldots,\bolda_{i_n}))
% \end{align*}
% \end{lemma}
% \begin{proof}
% This can be proved by induction on~$m$ and using the linearity of determinant in columns
% \begin{align*}
% &\det ((\boldb_1,\ldots,\boldb_{i-1},\boldb_i+\boldc,\boldb_{i+1},\ldots,\boldb_n))\\
% = &\det ((\boldb_1,\ldots,\boldb_{i-1},\boldb_i,\boldb_{i+1},\ldots,\boldb_n))
% + \det ((\boldb_1,\ldots,\boldb_{i-1},\boldc,\boldb_{i+1},\ldots,\boldb_n))
% \end{align*}
% for column vectors $\boldb_1,\ldots,\boldb_n,\boldc\in\mathbb{R}^{n\times 1}$.
% \end{proof}
Lemma~\ref{lemma:01indicator} is proved in Section~\ref{section:01indicator}, and its more involved counterpart Lemma~\ref{lemma:10indicator} is proved in Section~\ref{section:10indicator}. Finally, Lemma~\ref{lemma:g} and Lemma~\ref{lemma:g+} are proved in Section~\ref{section:gProofs}.

\section{Proof of Lemma~\ref{lemma:01indicator}}\label{section:01indicator}
% We now show that % Lemma~\ref{lemma:01indicator}
% for~$\boldc$ and~$\boldc'$ in~$\{0,1\}^n$ such that~$\boldc\in B_2(\boldc')$, if~$\1_{10}(\boldc)=\1_{10}(\boldc')$ and $h(\boldc)=h(\boldc')$, then $\1_{01}(\boldc)=\1_{01}(\boldc')$.
We now show that % Lemma~\ref{lemma:01indicator}
for any~$\boldc$ and~$\boldc'$ in~$\{0,1\}^n$ that satisfy~$\boldc\in B_2(\boldc')$, if~$\1_{10}(\boldc)=\1_{10}(\boldc')$ and $h(\boldc)=h(\boldc')$ (see~\eqref{equation:fgfunction} for definition of the $h$ function), then $\1_{01}(\boldc)=\1_{01}(\boldc')$.
Since~$\boldc$ and~$\boldc'$ have an identical~$10$-indicator, they can be written as
%Let~$\ell$ be the number of~$1$'s in~$\1_{10}(\boldc)$, and let
\begin{align}\label{cform}
\boldc&=0^{\pi_0}1^{\pi_1}0^{\pi_2}1^{\pi_3}\cdots 0^{\pi_{2\ell}}1^{\pi_{2\ell+1}},\nonumber\\
\boldc'&=0^{\tau_0}1^{\tau_1}0^{\tau_2}1^{\tau_3}\cdots 0^{\tau_{2\ell}}1^{\tau_{2\ell+1}},
\end{align}
where~$\{\pi_i\}_{i=0}^{2\ell+1}$ and~$\{\tau_i\}_{i=0}^{2\ell+1}$ are nonnegative integers such that~$\pi_i$ and~$\tau_i$ are strictly positive for every~$i\notin \{0,2\ell+1\}$, and such that~$\pi_{2i}+\pi_{2i+1}=\tau_{2i}+\tau_{2i+1}$ for all~$i\in\{0,1,\ldots,\ell\}$. In addition, since~$h(\boldc)_1=h(\boldc')_1$ it follows from Lemma~\ref{lemma:equivMod3} that~$\1_{01}(\boldc) \cdot \1 =\1_{01}(\boldc') \cdot \1$. %the number of~$1$'s in the~$01$-indicators is equal. 
Hence, we have
\begin{align*}
	\1_{01}(\boldc) \cdot \1 =\1_{01}(\boldc') \cdot \1=\ell+1   & \mbox{ if }&\pi_0>0,~\pi_{2\ell+1}>0\\
	\1_{01}(\boldc) \cdot \1 =\1_{01}(\boldc') \cdot \1=\ell   & \mbox{ if }&\pi_0>0,~\pi_{2\ell+1}\le 0\\
    &~&\mbox{ or }&\pi_0= 0,~\pi_{2\ell+1}>0 \\
	\1_{01}(\boldc) \cdot \1 =\1_{01}(\boldc') \cdot \1=\ell-1   & \mbox{ if }&\pi_0<0,~\pi_{2\ell+1}<0\\
\end{align*}
if~$\pi_0$ and~$\pi_{2\ell+1}$ (resp.~$\tau_0$ and~$\tau_{2\ell+1}$) are both positive then this number is~$\ell+1$, if precisely one of them is positive then it is~$\ell$, and if they are both zero it is~$\ell-1$.

Let~$\boldd=0^{\gamma_0}1^{\gamma_1}0^{\gamma_2}1^{\gamma_3}\cdots 0^{\gamma_{2\ell}}1^{\gamma_{2\ell+1}}\in\{0,1\}^{n-2}$ be a common subsequence of $\boldc$ and $\boldc'$ which is obtained by deleting two bits from either~$\boldc$ or~$\boldc'$, where~$\gamma_i\ge 0$ for all~$i$. Then, it is readily verified that
\begin{align*}
\sum^{2\ell+1}_{i=0}(\pi_i-\gamma_i) &= 2,~\sum^{2\ell+1}_{i=0}(\tau_i-\gamma_i) = 2,\mbox{, and hence}\\
%\end{align*}
%and hence
%\begin{align*}
\sum^{2\ell+1}_{i=1}|\pi_i-\tau_i| &\le \sum^{2\ell+1}_{i=1}|\pi_i-\gamma_i|+\sum^{2\ell+1}_{i=1}|\tau_i-\gamma_i| = 4.
\end{align*}
Moreover, since $\pi_{2i}+\pi_{2i+1}=\tau_{2i}+\tau_{2i+1}$ for all~$i\in\{0,1,\ldots,\ell\}$, it follows that $|\pi_{2i}-\tau_{2i}|=|\pi_{2i+1}-\tau_{2i+1}|$. Assuming for contradiction that the~$01$-indicators do not coincide implies either of the following cases.

\textbf{Case~(a)}. There exists an integer~$j\in [\ell]$ such that $|\pi_{2j}-\tau_{2j}|$ is either~$1$ or~$2$ and~$\pi_{2i}=\tau_{2i}$ for~$i\ne j$.

\textbf{Case~(b)}. There exist two integers~$m$ and~$r$ (where~$m<r$) such that $|\pi_{2m}-\tau_{2m}|=|\pi_{2r}-\tau_{2r}|=1$, and~$\pi_{2i}=\tau_{2i}$ for~$i\notin\{m,r\}$. 

In Case (a), since~$\pi_{2i}+\pi_{2i+1}=\tau_{2i}+\tau_{2i+1}$ for every~$i$ and~$\pi_{2i}=\tau_{2i}$ for every~$i\ne j$, it follows that~$\1_{01}(\boldc)$ and~$\1_{01}(\boldc')$ differ in precisely two positions~$s$ and~$t$ such that~$1\le s-t\le 2$. 
%Moreover, since~$h(\boldc)_1=h(\boldc')_1$, Lemma~\ref{lemma:equivMod3} implies that the number of~$1$'s in~$\1_{01}(\boldc)$ and in~$\1_{01}(\boldc')$ is equal. 
Hence, since the number of~$1$'s in the~$01$-indicators is equal, it follows that~$\1_{01}(\boldc )_{s}=\1_{01}(\boldc')_{t}$, $\1_{01}(\boldc)_{t}=\1_{01}(\boldc')_{s}$, and $\1_{01}(\boldc)_{s}\ne\1_{01}(\boldc)_{t}$, and therefore
\begin{alignat}{2}\label{equation:hc}
    h(\boldc)_2-h(\boldc')_2&=&&\;(\1_{01}(\boldc)_{s}-\1_{01}(\boldc')_{s}){s+1\choose 2}+\nonumber\\
    &~&&\;(\1_{01}(\boldc)_{t}-\1_{01}(\boldc')_{t}){t+1\choose 2}\nonumber\\
    &=&&\pm\left({s+1\choose 2}-{t+1\choose 2}  \right).
\end{alignat}
Since~$1\le s-t\le 2$, it follows that~\eqref{equation:hc} equals either~$\pm(t+1)$ or~$\pm(2t+3)$, and a contradiction follows since neither of which is~$0$ modulo~$2n$, .

Similarly, in Case~(b), if non of~$\pi_{2m},\tau_{2m},\pi_{2m+1},\tau_{2m+1},$ $\pi_{2r},\tau_{2r},\pi_{2r+1},\tau_{2r+1}$ is zero, then~$\1_{01}(\boldc)$ and~$\1_{01}(\boldc')$ differ in four positions~$s,s+1,t$, and~$t+1$, and hence
\begin{alignat}{2}\label{equation:Caseb}
    h(\boldc)_2-h(\boldc')_2&=&&\;(\1_{01}(\boldc)_{s}-\1_{01}(\boldc')_{s}){s+1\choose 2}+\nonumber\\
    &~&&\;(\1_{01}(\boldc)_{s+1}-\1_{01}(\boldc')_{s+1}){s+2\choose 2}+\nonumber\\
    &~&&\;(\1_{01}(\boldc)_{t}-\1_{01}(\boldc')_{t}){t+1\choose 2}+\nonumber\\
    &~&&\;(\1_{01}(\boldc)_{t+1}-\1_{01}(\boldc')_{t+1}){t+2\choose 2}.
\end{alignat}
Once again, since~$\1_{01}(\boldc)$ and~$\1_{01}(\boldc')$ have an identical number of~$1$'s, we have that
\begin{align*}
    \1_{01}(\boldc)_s&=\1_{01}(\boldc')_{s+1}&
    \1_{01}(\boldc)_{s+1}&=\1_{01}(\boldc')_s\\
    \1_{01}(\boldc)_{t}&=\1_{01}(\boldc')_{t+1}&
    \1_{01}(\boldc)_{t+1}&=\1_{01}(\boldc')_{t}\\
    \1_{01}(\boldc)_s&\ne\1_{01}(\boldc')_{s}&
    \1_{01}(\boldc)_t&\ne\1_{01}(\boldc')_{t}.
\end{align*}
This readily implies that~\eqref{equation:Caseb} equals either~$\pm(s-t)$ or~$\pm(s+t+2)$, and since non of which is~$0$ modulo~$2n$, another contradiction is obtained.
%, which implies that the vectors~$\1_{01}(\boldc)$ and~$\1_{01}(\boldc')$ are equal.
If~$\pi_{2m}=0$ (resp.~$\tau_{2m}=0$), by the discussion after Eq.~\eqref{cform} it follows that~$\tau_{2r+1}=0$ (resp.~$\pi_{2r+1}=0$), and hence~$\1_{01}(\boldc)$ and~$\1_{01}(\boldc')$ differ in the first and last positions. Hence,~\eqref{equation:Caseb} becomes~$\pm(1-\frac{n(n-1)}{2})$, which is nonzero modulo~$2n$, and the claim follows.

\section{Proof of Lemma~\ref{lemma:10indicator}}\label{section:10indicator}
Since~$c\in B_2(c')$ it follows that there exist integers~$i_1,i_2,j_1$, and~$j_2$ such that
\begin{align*}
		c &\overset{\mbox{\tiny{del' }}i_1}{\longrightarrow}d  \overset{\mbox{\tiny{del' }}j_1}{\longrightarrow}e\\
		c'&\overset{\mbox{\tiny{del' }}i_2}{\longrightarrow}d' \overset{\mbox{\tiny{del' }}j_2}{\longrightarrow}e
\end{align*}
and by Lemma~\ref{lemma:twoDeletions10} it follows that there exist integers~$\ell_1,\ell_2,k_1$, and~$k_2$ such that
\begin{align*}
	\1_{10}(c) &\overset{\mbox{\tiny{del' }}\ell_1}{\longrightarrow}\1_{10}(d) \overset{\mbox{\tiny{del' }}k_1}{\longrightarrow}\1_{10}(e)\\
	\1_{10}(c')&\overset{\mbox{\tiny{del' }}\ell_2}{\longrightarrow}\1_{10}(d')\overset{\mbox{\tiny{del' }}k_2}{\longrightarrow}\1_{10}(e).
\end{align*}

Due to symmetry between~$\boldc$ and~$\boldc'$, we distinguish between the following three cases. In each case, the difference between the~$f$ values of~$\boldc$ and~$\boldc'$ are given in terms of the function~$g$ (Eq.~\eqref{equation:g}). Further, the computation of these three differences, which is tedious but straightforward, is deferred to the appendices. %Notice that the equalities below are modular, and yet ordinary equality holds due to trivial bounds on~$g$.

\textbf{Case~(a)}. If $\ell_1\le \ell_2 < k_2 \le k_1$Ã¯Â¼Å’ (Fig.~\ref{figure:Casea}), %meaning that there are two disjoint segment shifts in the same direction,
then
\begin{align*}
	\1_{10}(\boldc)_{t} &= \1_{10}(\boldc')_t   & \mbox{ if }&t<\ell_1\\
    &~&\mbox{ or }&\ell_2< t< k_2 \\
    &~&\mbox{ or }& t> k_1,\\
	\1_{10}(\boldc)_{t+1} &= \1_{10}(\boldc')_t   & \mbox{ if }&\ell_1\le t\le \ell_2-1,\\
	\1_{10}(\boldc)_{t} &= \1_{10}(\boldc')_{t+1} & \mbox{ if }&k_2\le t\le k_1-1,
\end{align*}
Thus, for~$e\in\{0,1,2\}$,
    \begin{alignat}{2}\label{equation:g_m-Case(a)}
    	%(f(\boldc)-f(\boldc'))_e
    	(\mathbbm{1}_{10}(\boldc)-\mathbbm{1}_{10}(\boldc'))\cdot \boldm^{(e)}
        &=\;& g_{\boldm^{(e)}}(\ell_1,(\1_{10}(\boldc)^{(\ell_1,\ell_2)},\1_{10}(\boldc')_{\ell_2}))-\nonumber\\ &~&g_{\boldm^{(e)}}(k_2,(\1_{10}(\boldc')^{(k_2,k_1)},\1_{10}(\boldc)_{k_1})).
    \end{alignat}

\textbf{Case~(b)}. If~$\ell_1\le \ell_2< k_1\le k_2$ (Fig.~\ref{figure:Caseb}), %meaning that there are two disjoint segment shifts in opposite direction,
then
\begin{align*}
	\1_{10}(\boldc)_{t} &= \1_{10}(\boldc')_t & \mbox{ if }&t<l_1\\
	&~&\mbox{ or }&l_2<t<k_1.\\
    &~&\mbox{ or }&t>k_2.\\
	\1_{10}(\boldc)_{t+1} &= \1_{10}(\boldc')_t & \mbox{ if }&\ell_1\le t\le \ell_2-1\\
	&~&\mbox{ or }&k_1\le t\le k_2-1.
\end{align*}
Thus, for~$e\in\{0,1,2\}$,
\begin{alignat}{2}\label{equation:g_m-Case(b)}
%(f(\boldc)-f(\boldc'))_e
(\mathbbm{1}_{10}(\boldc)-\mathbbm{1}_{10}(\boldc'))\cdot \boldm^{(e)}
&=\;&
g_{\boldm^{(e)}}(\ell_1,(\1_{10}(\boldc)^{(\ell_1,\ell_2)},\1_{10}(\boldc')_{\ell_2}))+\nonumber\\
&\phantom{=}&g_{\boldm^{(e)}}(k_1,(\1_{10}(\boldc)^{(k_1,k_2)},\1_{10}(\boldc')_{k_2})).
\end{alignat}

\textbf{Case~(c)}. If~$\ell_1<k_1\le \ell_2<k_2$ (Fig.~\ref{figure:Casec}), %meaning that the two segment shifts over lap,
then
\begin{align*}
	\1_{10}(\boldc)_{t} &= \1_{10}(\boldc')_t &\mbox{ if }& t<l_1\\
	&~&\mbox{ or }&t>k_2,\\
	\1_{10}(\boldc)_{t+1} &= \1_{10}(\boldc')_t &\mbox{ if }& \ell_1 \le t\le k_1-2\\
	&~&\mbox{ or }&\ell_2+1\le t\le k_2-1,\\
	\1_{10}(\boldc)_{t+2} &= \1_{10}(\boldc')_t &\mbox{ if }& k_1-1\le t\le \ell_2-1.
\end{align*}
Thus, for~$e\in\{0,1,2\}$,
\begin{alignat}{3}\label{equation:g_m-Case(c)}
	%(f(\boldc)-f(\boldc'))_e\equiv_{n^e}
    (\mathbbm{1}_{10}(\boldc)-\mathbbm{1}_{10}(\boldc'))\cdot \boldm^{(e)}=\;&
	 g_{\boldm^{(e)}}(&&\ell_1,(\1_{10}(\boldc)^{(\ell_1,k_1-1)},\nonumber\\
	 &~&&\1_{10}(\boldc)^{(k_1+1,\ell_2+1)},\1_{10}(\boldc')_{\ell_2}))+\nonumber\\
	&g_{\boldm^{(e)}}(&&k_1, (\1_{10}(\boldc)^{(k_1,k_2)},\1_{10}(\boldc')_{k_2})).
\end{alignat}

\begin{figure} %%% Case (a) %%%
	\begin{tikzpicture}[line cap=round,line join=round,>=triangle 45,x=1cm,y=1cm]
\clip(-1,0) rectangle (16,4);
\fill[line width=2pt,fill=black,fill opacity=0.1] (1,1) -- (15,1) -- (15,3) -- (1,3) -- cycle;

\draw [line width=1pt] (1,2)-- (15,2);
\draw [line width=2pt] (1,1)-- (15,1);
\draw [line width=2pt] (15,1)-- (15,3);
\draw [line width=2pt] (15,3)-- (1,3);
\draw [line width=2pt] (1,3)-- (1,1);
\draw [line width=2pt] (3,3)-- (3,1);
\draw [line width=2pt] (4,3)-- (4,1);
\draw [line width=2pt] (6,3)-- (6,1);
\draw [line width=2pt] (7,3)-- (7,1);
\draw [line width=2pt] (9,3)-- (9,1);
\draw [line width=2pt] (10,1)-- (10,3);
\draw [line width=2pt] (12,3)-- (12,1);
\draw [line width=2pt] (13,3)-- (13,1);

\draw (0,2.5) node[anchor=center] {$\1_{10}(c)$};
\draw (0,1.5) node[anchor=center] {$\1_{10}(c')$};
\draw (6.5,1.5) node[anchor=center] {$\star$};
\draw (3.5,2.5) node[anchor=center] {$\star$};
\draw (12.5,2.5) node[anchor=center] {$\star$};
\draw (9.5,1.5) node[anchor=center] {$\star$};
\draw [line width=2pt] (3.5,1.5)-- (4.5,2.5);
\draw [line width=2pt] (4.5,1.5)-- (5.5,2.5);
\draw [line width=2pt] (5.5,1.5)-- (6.5,2.5);
\draw [line width=2pt] (9.5,2.5)-- (10.5,1.5);
\draw [line width=2pt] (10.5,2.5)-- (11.5,1.5);
\draw [line width=2pt] (11.5,2.5)-- (12.5,1.5);

\draw (2,1.5) node[anchor=center] {$=$};
\draw (8,1.5) node[anchor=center] {$=$};
\draw (14,1.5) node[anchor=center] {$=$};

\draw (2,2.5) node[anchor=center] {$=$};
\draw (8,2.5) node[anchor=center] {$=$};
\draw (14,2.5) node[anchor=center] {$=$};

\draw (3.5,3.5) node[anchor=center] {$\ell_1$};
\draw (6.5,3.5) node[anchor=center] {$\ell_2$};
\draw (12.5,3.5) node[anchor=center] {$k_1$};
\draw (9.5,3.5) node[anchor=center] {$k_2$};
\end{tikzpicture}\caption{Case (a)}\label{figure:Casea}
\end{figure}

\begin{figure} %%% Case (b) %%%
	\begin{tikzpicture}[line cap=round,line join=round,>=triangle 45,x=1cm,y=1cm]
	\clip(-1,0) rectangle (16,4);
	\fill[line width=2pt,fill=black,fill opacity=0.1] (1,1) -- (15,1) -- (15,3) -- (1,3) -- cycle;
	
	\draw [line width=1pt] (1,2)-- (15,2);
	\draw [line width=2pt] (1,1)-- (15,1);
	\draw [line width=2pt] (15,1)-- (15,3);
	\draw [line width=2pt] (15,3)-- (1,3);
	\draw [line width=2pt] (1,3)-- (1,1);
	\draw [line width=2pt] (3,3)-- (3,1);
	\draw [line width=2pt] (4,3)-- (4,1);
	\draw [line width=2pt] (6,3)-- (6,1);
	\draw [line width=2pt] (7,3)-- (7,1);
	\draw [line width=2pt] (9,3)-- (9,1);
	\draw [line width=2pt] (10,1)-- (10,3);
	\draw [line width=2pt] (12,3)-- (12,1);
	\draw [line width=2pt] (13,3)-- (13,1);
	
	\draw (0,2.5) node[anchor=center] {$\1_{10}(c)$};
	\draw (0,1.5) node[anchor=center] {$\1_{10}(c')$};
	\draw (6.5,1.5) node[anchor=center] {$\star$};
	\draw (3.5,2.5) node[anchor=center] {$\star$};
	\draw (12.5,1.5) node[anchor=center] {$\star$};
	\draw (9.5,2.5) node[anchor=center] {$\star$};
	
	\draw [line width=2pt] (3.5,1.5)-- (4.5,2.5);
	\draw [line width=2pt] (4.5,1.5)-- (5.5,2.5);
	\draw [line width=2pt] (5.5,1.5)-- (6.5,2.5);
	
	\draw [line width=2pt] (9.5,1.5)-- (10.5,2.5);
	\draw [line width=2pt] (10.5,1.5)-- (11.5,2.5);
	\draw [line width=2pt] (11.5,1.5)-- (12.5,2.5);
	
	\draw (2,1.5) node[anchor=center] {$=$};
	\draw (8,1.5) node[anchor=center] {$=$};
	\draw (14,1.5) node[anchor=center] {$=$};
	
	\draw (2,2.5) node[anchor=center] {$=$};
	\draw (8,2.5) node[anchor=center] {$=$};
	\draw (14,2.5) node[anchor=center] {$=$};
	
	\draw (3.5,3.5) node[anchor=center] {$\ell_1$};
	\draw (6.5,3.5) node[anchor=center] {$\ell_2$};
	\draw (12.5,3.5) node[anchor=center] {$k_2$};
	\draw (9.5,3.5) node[anchor=center] {$k_1$};
	\end{tikzpicture}\caption{Case (b)}\label{figure:Caseb}
\end{figure}

\begin{figure}[h] %%% Case (c) %%%
	\begin{tikzpicture}[line cap=round,line join=round,>=triangle 45,x=1cm,y=1cm]
	\clip(-1,0) rectangle (16,4);
	\fill[line width=2pt,fill=black,fill opacity=0.1] (1,1) -- (15,1) -- (15,3) -- (1,3) -- cycle;
	
	\draw [line width=1pt] (1,2)-- (15,2);
	\draw [line width=2pt] (1,1)-- (15,1);
	\draw [line width=2pt] (15,1)-- (15,3);
	\draw [line width=2pt] (15,3)-- (1,3);
	\draw [line width=2pt] (1,3)-- (1,1);
	\draw [line width=2pt] (3,3)-- (3,1);
	\draw [line width=2pt] (4,3)-- (4,1);
	\draw [line width=2pt] (6,3)-- (6,1);
	\draw [line width=2pt] (7,3)-- (7,1);
	\draw [line width=2pt] (9,3)-- (9,1);
	\draw [line width=2pt] (10,1)-- (10,3);
	\draw [line width=2pt] (12,3)-- (12,1);
	\draw [line width=2pt] (13,3)-- (13,1);
	
	\draw (0,2.5) node[anchor=center] {$\1_{10}(c)$};
	\draw (0,1.5) node[anchor=center] {$\1_{10}(c')$};
	\draw (6.5,2.5) node[anchor=center] {$\star$};
	\draw (3.5,2.5) node[anchor=center] {$\star$};
	\draw (12.5,1.5) node[anchor=center] {$\star$};
	\draw (9.5,1.5) node[anchor=center] {$\star$};
	
	\draw [line width=2pt] (3.5,1.5)-- (4.5,2.5);
	\draw [line width=2pt] (4.5,1.5)-- (5.5,2.5);
	\draw [line width=2pt] (5.5,1.5)-- (7.5,2.5);
	
	\draw [line width=2pt] (6.5,1.5)-- (8.5,2.5);
	\draw [line width=2pt] (7.5,1.5)-- (9.5,2.5);
	
	\draw [line width=2pt] (8.5,1.5)-- (10.5,2.5);
	
	\draw [line width=2pt] (10.5,1.5)-- (11.5,2.5);
	\draw [line width=2pt] (11.5,1.5)-- (12.5,2.5);
	
	\draw (2,1.5) node[anchor=center] {$=$};
	\draw (14,1.5) node[anchor=center] {$=$};
	
	\draw (2,2.5) node[anchor=center] {$=$};
	\draw (14,2.5) node[anchor=center] {$=$};
	
	\draw (3.5,3.5) node[anchor=center] {$\ell_1$};
	\draw (6.5,3.5) node[anchor=center] {$k_1$};
	\draw (12.5,3.5) node[anchor=center] {$k_2$};
	\draw (9.5,3.5) node[anchor=center] {$\ell_2$};
	\end{tikzpicture}\caption{Case (c)}\label{figure:Casec}
\end{figure}

Note that if~$f(\boldc)=f(\boldc')$, then~$\1_{10}(\boldc) \cdot \boldm^{(e)}\equiv \1_{10}(\boldc) \cdot \boldm^{(e)}\bmod n_e$, where~$n_0=2n,n_1=n^2,$ and~$n_2=n^3$.
Hence, from~\eqref{equation:g_m-Case(a)}-\eqref{equation:g_m-Case(c)} we have that
\begin{alignat}{2}\label{equation:gfunctionmodular}
&g_{\boldm^{(e)}}(\ell_1,(\1_{10}(\boldc)^{(\ell_1,\ell_2)},\1_{10}(\boldc')_{\ell_2}))- g_{\boldm^{(e)}}(k_2,(\1_{10}(\boldc')^{(k_2,k_1)},\1_{10}(\boldc)_{k_1}))&& \equiv 0 \bmod 2n,\nonumber\\
&g_{\boldm^{(e)}}(\ell_1,(\1_{10}(\boldc)^{(\ell_1,\ell_2)},\1_{10}(\boldc')_{\ell_2}))
+
g_{\boldm^{(e)}}(k_1,(\1_{10}(\boldc)^{(k_1,k_2)},\1_{10}(\boldc')_{k_2}))&& \equiv 0 \bmod n^2,\mbox{ and}\nonumber\\
&g_{\boldm^{(e)}}(\ell_1,(\1_{10}(\boldc)^{(\ell_1,k_1-1)},
	 \1_{10}(\boldc)^{(k_1+1,\ell_2+1)},\1_{10}(\boldc')_{\ell_2}))\nonumber\\
&+
	g_{\boldm^{(e)}}(k_1, (\1_{10}(\boldc)^{(k_1,k_2)},\1_{10}(\boldc')_{k_2}))&& \equiv 0 \bmod n^3.
g\end{alignat}
%for Case~(a), Case~(b), and Case~(c) respectively.
In what follows, we show that these equalities also hold in their non modular version.
On the other hand, we have
\begin{align*}
-\boldm^{(e)}_{r+k-2}\le g_{\boldm^{(e)}}(r,\boldx)\le \boldm^{(e)}_{r+k-2}
\end{align*}
for any~$\boldx\in\{0,1\}^{n-1}$ and any integer~$r$ that satisfies~$r+k-2\le n-1$. Therefore,
\begin{align}\label{equation:gfunctionbound}
-\boldm^{(e)}_{\ell_2}-\boldm^{(e)}_{k_2}\le &g_{\boldm^{(e)}}(\ell_1,(\1_{10}(\boldc)^{(\ell_1,\ell_2)},\1_{10}(\boldc')_{\ell_2}))- g_{\boldm^{(e)}}(k_2,(\1_{10}(\boldc')^{(k_2,k_1)},\1_{10}(\boldc)_{k_1}))\le \boldm^{(e)}_{\ell_2}+\boldm^{(e)}_{k_2},\nonumber\\
-\boldm^{(e)}_{\ell_2}-\boldm^{(e)}_{k_2}\le &g_{\boldm^{(e)}}(\ell_1,(\1_{10}(\boldc)^{(\ell_1,\ell_2)},\1_{10}(\boldc')_{\ell_2}))+
g_{\boldm^{(e)}}(k_1,(\1_{10}(\boldc)^{(k_1,k_2)},\1_{10}(\boldc')_{k_2}))\le \boldm^{(e)}_{\ell_2}+\boldm^{(e)}_{k_2},\mbox{ and}\nonumber\\
-\boldm^{(e)}_{\ell_2}-\boldm^{(e)}_{k_2}\le &g_{\boldm^{(e)}}(\ell_1,(\1_{10}(\boldc)^{(\ell_1,k_1-1)},
	 \1_{10}(\boldc)^{(k_1+1,\ell_2+1)},\1_{10}(\boldc')_{\ell_2}))\nonumber\\
     &+
	g_{\boldm^{(e)}}(k_1, (\1_{10}(\boldc)^{(k_1,k_2)},\1_{10}(\boldc')_{k_2}))\le \boldm^{(e)}_{\ell_2}+\boldm^{(e)}_{k_2}.
\end{align}
Further note that
\begin{align}\label{equation:momentbound}
\boldm^{(0)}_{\ell_2}+\boldm^{(0)}_{k_2}<2n,\boldm^{(1)}_{\ell_2}+\boldm^{(1)}_{k_2}<n^2,\boldm^{(2)}_{\ell_2}+\boldm^{(2)}_{k_2}<n^3
\end{align}
Combining~\eqref{equation:gfunctionmodular},~\eqref{equation:gfunctionbound}, and~\eqref{equation:momentbound}, we conclude that if~$f(\boldc)=f(\boldc')$, then
\begin{align}\label{equation:gfunctioncase-a}
&g_{\boldm^{(e)}}(\ell_1,(\1_{10}(\boldc)^{(\ell_1,\ell_2)},\1_{10}(\boldc')_{\ell_2}))- g_{\boldm^{(e)}}(k_2,(\1_{10}(\boldc')^{(k_2,k_1)},\1_{10}(\boldc)_{k_1}))=0,\\
\label{equation:gfunctioncase-b}
&g_{\boldm^{(e)}}(\ell_1,(\1_{10}(\boldc)^{(\ell_1,\ell_2)},\1_{10}(\boldc')_{\ell_2}))+
g_{\boldm^{(e)}}(k_1,(\1_{10}(\boldc)^{(k_1,k_2)},\1_{10}(\boldc')_{k_2}))=0,\mbox{ and}\\
\label{equation:gfunctioncase-c}
&g_{\boldm^{(e)}}(\ell_1,(\1_{10}(\boldc)^{(\ell_1,k_1-1)},
	 \1_{10}(\boldc)^{(k_1+1,\ell_2+1)},\1_{10}(\boldc')_{\ell_2}))+
	g_{\boldm^{(e)}}(k_1, (\1_{10}(\boldc)^{(k_1,k_2)},\1_{10}(\boldc')_{k_2}))=0.
\end{align}
%respectively for Case~(a), Case~(b), and Case~(c).

For Case~(a), Equation~\eqref{equation:gfunctioncase-a} and Lemma~\ref{lemma:g}
implies that
\begin{align*}
    \1_{10}(\boldc)_{\ell_1}&=\ldots=\1_{10}(\boldc)_{\ell_2}=\1_{10}(\boldc')_{\ell_2}\\
    \1_{10}(\boldc')_{k_2}&=\ldots=\1_{10}(\boldc')_{k_1}=\1_{10}(\boldc)_{k_1},
\end{align*}
which readily implies that
\begin{align*}
    \1_{10}(\boldc')_{t}=\1_{10}(\boldc)_{t+1}=\1_{10}(\boldc)_{t}
\end{align*}
for~$\ell_1\le t<\ell_2$ and
\begin{align*}
    \1_{10}(\boldc)_{t}&=\1_{10}(\boldc')_{t+1}=\1_{10}(\boldc')_{t}
\end{align*}
for~$k_2\le t<k_1$. Together with~$\1_{10}(\boldc)_{\ell_2}=\1_{10}(\boldc')_{\ell_2}$ and~$\1_{10}(\boldc')_{k_1}=\1_{10}(\boldc)_{k_1}$, we have that~$\1_{10}(\boldc)=\1_{10}(\boldc')$.

%Similarly to Case~(a),
For Case~(b), Equation~\eqref{equation:gfunctioncase-b} and Lemma~\ref{lemma:g} implies that
\begin{align*}
    \1_{10}(\boldc)_{\ell_1}&=\ldots=\1_{10}(\boldc)_{\ell_2}=\1_{10}(\boldc')_{\ell_2}\\
    \1_{10}(\boldc')_{k_1}&=\ldots=\1_{10}(\boldc')_{k_2}=\1_{10}(\boldc)_{k_2}
\end{align*}
and hence
\begin{align*}
    \1_{10}(\boldc')_{t}=\1_{10}(\boldc)_{t+1}=\1_{10}(\boldc)_{t}
\end{align*}
for~$\ell_1\le t<\ell_2$ and~$k_1\le t<k_2$. $\1_{10}(\boldc)=\1_{10}(\boldc')$.

For Case~(c), Equation~\eqref{equation:gfunctioncase-c} and Lemma~\ref{lemma:g+} imply that either
\begin{align}\label{equation:CasecOptionA}
    &\1_{10}(\boldc)_{\ell_1}=\ldots=\1_{10}(\boldc)_{k_2}=\1_{10}(\boldc')_{\ell_2}=\1_{10}(\boldc')_{k_2}%\mbox{, or}\\
\end{align}
or
\begin{align}
\label{equation:CasecOptionB}
&\1_{10}(\boldc)_{\ell_1}=\ldots=\1_{10}(\boldc)_{k_1-1}=\1_{10}(\boldc)_{k_1+1},\nonumber\\
&\1_{10}(\boldc)_{i}+\1_{10}(\boldc)_{i+1}=1\mbox{ for }i\in\{k_1,\ldots,\ell_2\},\nonumber\\
&\1_{10}(\boldc')_{\ell_2}+\1_{10}(\boldc')_{k_2}=1,\mbox{ and}\nonumber\\
&\1_{10}(\boldc)_{\ell_2+1}=\ldots=\1_{10}(\boldc)_{k_2}=\1_{10}(\boldc')_{k_2}.
\end{align}
If~\eqref{equation:CasecOptionA} is true, we can obtain~$\boldc = \boldc'$ by following similar steps as above.

If~\eqref{equation:CasecOptionB} is true,
we have
\begin{align*}
    \1_{10}(\boldc')_{t}=\1_{10}(\boldc)_{t+1}=\1_{10}(\boldc)_{t}
\end{align*}
for~$\ell_1\le t\le k_1-2$ and~$\ell_2+1\le t\le k_2-1$. Further more, we have
\begin{align*}
    \1_{10}(\boldc')_{t}=\1_{10}(\boldc)_{t+2}=1-\1_{10}(\boldc)_{t+1}=\1_{10}(\boldc)_{t}
\end{align*}
for~$k_1\le t\le \ell_2-1$. In addition, we have~$\1_{10}(\boldc')_{k_1-1}=\1_{10}(\boldc)_{k_1+1}=\1_{10}(\boldc)_{k_1-1}$,~$\1_{10}(\boldc')_{\ell_2}=1-\1_{10}(\boldc')_{\ell_2+1}=\1_{10}(\boldc)_{\ell_2}$ and~$\1_{10}(\boldc')_{k_2}=\1_{10}(\boldc)_{k_2}$. Therefore, we conclude that $\boldc = \boldc'$.

\section{Proofs of~$g$-lemmas}\label{section:gProofs}

\begin{proof}(of Lemma~\ref{lemma:g})
According to Eq. \eqref{equation:opposite}, if~$\lambda=1$, then Eq. \eqref{equation:lemmagEq} can be written as 
     \begin{align*}
        g_{\boldm^{(0)}}(r_1,\boldx)- g_{\boldm^{(0)}}(r_2,\overline{\boldy})&=0\mbox{, and}\\
        g_{\boldm^{(1)}}(r_1,\boldx)- g_{\boldm^{(1)}}(r_2,\overline{\boldy})&=0.
    \end{align*}
Therefore, it suffices to prove the claim for ~$\lambda=-1$.    
We distinguish between four cases according to the value of~$(y_1,y_{s_2})$.

%\begin{description}
\textbf{Case~$(1)$.} $(y_1,y_{s_2})=(0,1)$
\hfill

we have that
\begin{align*}
        & g_{\boldm^{(e)}}(r_1,\boldx)-g_{\boldm^{(e)}}(r_2,\boldy)\\
        =& \boldm^{(e)}_{r_1}x_1+\sum^{s_1-1}_{t=2}(\boldm^{(e)}_{t+r_1-1}-\boldm^{(e)}_{t+r_1-2})x_t -\\
        &\boldm^{(e)}_{r_1+s_1-2}x_{s_1}
        - \boldm^{(e)}_{r_2}y_1-\sum^{s_2-1}_{t=2}(\boldm^{(e)}_{t+r_2-1}-\boldm^{(e)}_{t+r_2-2})y_t + \boldm^{(e)}_{r_2+s_2-2}y_{s_2}\\
        \ge & -\boldm^{(e)}_{r_1+s_1-2}-\sum^{s_2-1}_{t=2}(\boldm^{(e)}_{t+r_2-1}-\boldm^{(e)}_{t+r_2-2})+\;\boldm^{(e)}_{r_2+s_2-2}\\
        =&\boldm^{(e)}_{r_2}-\boldm^{(e)}_{r_1+s_1-2}>0,
    \end{align*}
a contradiction.

\textbf{Case~$(2)$.} $(y_1,y_{s_2})=(1,0)$
\hfill

From Lemma~\ref{lemma:gproperty} and \eqref{equation:lemmagEq} we have $g_{\boldm^{(e)}}(r_1,\overline{\boldx})+ g_{\boldm^{(e)}}(r_2,\overline{\boldy})=0$ for~$e\in\{0,1\}$, where~$\overline{\boldx}\triangleq\1-\boldx$ and~$\overline{\boldy}\triangleq\1-\boldy$. Since~$(\overline{y}_1,\overline{y}_{s_2})=(1,0)$, from the previous case we have that~$\overline{\boldx}$ and~$\overline{\boldy}$ are constant vectors. So are~$\boldx$ and~$\boldy$.

\textbf{Case~$(3)$.} $(y_1,y_{s_2})=(1,1)$
\hfill

Let
\begin{align*}
    S_1&\triangleq\{j:y_{j-r_2+1}=1,r_2+1\le j\le r_2+s_2-2\}\mbox{, and}\\
    S^c_1&\triangleq\{j:y_{j-r_2+1}=0,r_2+1\le j\le r_2+s_2-2\},
\end{align*}
and notice that
%Rewrite $g_{V,r_2}(Y)$ as
\begin{align}\label{case31}
g_{\boldm^{(0)}}(r_2,\boldy)&= \boldm^{(0)}_{r_2}-\boldm^{(0)}_{r_2+s_2-2}+ \sum^{s_2-1}_{j=2}(\boldm^{(0)}_{j+r_2-1}-\boldm^{(0)}_{j+r_2-2})y_j\nonumber\\
&= -\sum^{r_2+s_2-2}_{j=r_2+1}(\boldm^{(0)}_{j}-\boldm^{(0)}_{j-1})+ \sum^{s_2-1}_{j=2}(\boldm^{(0)}_{j+r_2-1}-\boldm^{(0)}_{j+r_2-2})y_j\nonumber\\
&=-\sum^{r_2+s_2-2}_{j=r_2+1}(\boldm^{(0)}_{j}-\boldm^{(0)}_{j-1})(1-y_j)\nonumber\\
&=-\sum_{j\in S^c_1}(\boldm^{(0)}_{j}-\boldm^{(0)}_{j-1})=-\sum_{j\in S^c_1}1,\mbox{ and similarly}\nonumber\\
g_{\boldm^{(1)}}(r_2,\boldy)&=-\sum_{j\in S^c_1}(\boldm^{(1)}_{j}-\boldm^{(1)}_{j-1})=-\sum_{j\in S^c_1}j.
\end{align}
% Similarly, we have
% \begin{equation}\label{case32}
% g_{\boldm^{(1)}}(r_2,\boldy)=-\sum_{j\in S^c_1}(\boldm^{(1)}_{j+1}-\boldm^{(1)}_{j})=-\sum_{j\in S^c_1}j.
% \end{equation}
Now, on the one hand if $x_{s_1}=0$ we have
\begin{align}\label{case33}
g_{\boldm^{(0)}}(r_1,\boldx)= \boldm^{(0)}_{r_1}x_1+\sum^{s_1-1}_{t=2}(\boldm^{(0)}_{t+r_1-1}-\boldm^{(0)}_{t+r_1-2})x_t\ge 0,
\end{align}
and hence,~\eqref{case31} and~\eqref{case33} imply that ${g_{\boldm^{(0)}}(r_1,\boldx) - g_{\boldm^{(0)}}(r_2,\boldy)\ge 0}$, and equality holds only when $g_{\boldm^{(0)}}(r_1,\boldx)$ and $g_{\boldm^{(0)}}(r_2,\boldy)$ are both~$0$, which by Lemma~\ref{lemma:VT} implies that~$\boldx$ and~$\boldy$ are constant vectors. On the other hand, if~$x_{s_1}=1$ let $S_2=\{j:x_{\max\{j-r_1+1,1\}}=0,1\le j\le r_1+s_1-2\}$, and notice that
\begin{alignat}{2}\label{case34}
g_{\boldm^{(0)}}(r_1,\boldx)&=&&\; \boldm^{(0)}_{r_1}x_1+\sum^{s_1-1}_{t=2}(\boldm^{(0)}_{t+r_1-1}-\boldm^{(0)}_{t+r_1-2})x_t-\boldm^{(0)}_{r_1+s_1-2}\nonumber\\
&=&&\; \boldm^{(0)}_{r_1}(x_1-1)+\sum^{s_1-1}_{t=2}(\boldm^{(0)}_{t+r_1-1}-\boldm^{(0)}_{t+r_1-2})(x_t-1)\nonumber\\
&=&&-\sum_{t\in S_2}1,\mbox{ and similarly}\nonumber\\
g_{\boldm^{(1)}}(r_1,\boldx)&=&&-\sum_{t\in S_2}t.
\end{alignat}
Inserting~\eqref{case31} and~\eqref{case34} into~\eqref{equation:lemmagEq}, we have
\begin{align*}%\label{case36}
-\sum_{t\in S_2}1+\sum_{j\in S^c_1}1&=0,\\
-\sum_{t\in S_2}t+\sum_{j\in S^c_1}j&=0.
\end{align*}
This implies that the sets~$S^c_1$ and~$S_2$ have the same cardinality and the same sum of elements. However, the maximum element in~$S_2$ is smaller than the minimum element in~$S^c_1$. Therefore~$S^c_1$ and~$S_2$ are empty, which implies that~$\boldx$ is the~$0$ vector and~$\boldy$ is the all~$1$'s vector.

\textbf{Case~$(4)$.} $(y_1,y_{s_2})=(0,0)$
\hfill

From Lemma~\ref{lemma:gproperty} and Eq. \eqref{equation:lemmagEq} we have $g_{\boldm^{(e)}}(r_1,\overline{\boldx})+ g_{\boldm^{(e)}}(r_2,\overline{\boldy})=0$ for~$e\in\{0,1\}$, where~$\overline{\boldx}\triangleq\1-\boldx$ and~$\overline{\boldy}\triangleq\1-\boldy$. Since~$(\overline{y}_1,\overline{y}_{s_2})=(1,1)$, from the previous case~$\overline{\boldx}$ and~$\overline{\boldy}$ are constant vectors, and thus so are~$\boldx$ and~$\boldy$.
\end{proof}

\begin{proof}(of Lemma~\ref{lemma:g+})
We distinguish between four cases according to the value of~$(x_{s_1+s_2+1},y_{s_2+s_3+1})$.

\textbf{Case~$(1)$.} $(x_{s_1+s_2+1},y_{s_2+s_3+1})=(0,0)$\hfill

Similar to~\eqref{case33}, we have that $g_{\boldm^{(0)}}(r_1,\boldx) + g_{\boldm^{(0)}}(r_2,\boldy)\ge 0$, where equality holds only if~$\boldx$ and~$\boldy$ are constant~$0$ vectors.

\textbf{Case~$(2)$.} $(x_{s_1+s_2+1},y_{s_2+s_3+1})=(1,1)$\hfill

From Lemma~\ref{lemma:gproperty} and Eq. \eqref{equation:lemmag+} we have $g_{\boldm^{(0)}}(r_1,\overline{\boldx})+ g_{\boldm^{(0)}}(r_2,\overline{\boldy})=0$. On the other hand, since~ $(\overline{x}_{s_1+s_2+1},\overline{y}_{s_2+s_3+1})=(0,0)$ , it follows that~$g_{\boldm^{(0)}}(r_1,\overline{\boldx}) + g_{\boldm^{(0)}}(r_2,\overline{\boldy})\ge 0$
where equality holds when~$\boldx$ and~$\boldy$ are constant~$1$ vectors.

\textbf{Case~$(3)$.} $(x_{s_1+s_2+1},y_{s_2+s_3+1})=(0,1)$\hfill

On the one hand, for~$y_1=0$ we have
\begin{alignat*}{2}
g_{\boldm^{(0)}}&(r_1&&,\boldx)+ g_{\boldm^{(0)}}(r_2,\boldy)=\\
&= &&\;\boldm^{(0)}_{r_1}x_1 + \sum^{s_1+1}_{t=2}(\boldm^{(0)}_{t+r_1-1}-\boldm^{(0)}_{t+r_1-2})x_t + \sum^{s_1+s_2-1}_{t=s_1+1}(\boldm^{(0)}_{t+r_1}-\boldm^{(0)}_{t+r_1-1})x_{t+1}\\
&~&& + \sum^{s_2}_{t=2}(\boldm^{(0)}_{t+r_2-1}-\boldm^{(0)}_{t+r_2-2})y_t+\sum^{s_2+s_3}_{t=s_2+1}(\boldm^{(0)}_{t+r_2-1}-\boldm^{(0)}_{t+r_2-2})y_t-\boldm^{(0)}_{r_2+s_2+s_3-1}\\
&= &&\;\boldm^{(0)}_{r_1}x_1 + \sum^{s_1+1}_{t=2}(\boldm^{(0)}_{t+r_1-1}-\boldm^{(0)}_{t+r_1-2})x_t  +\sum^{s_1+s_2-1}_{t=s_1+1}(\boldm^{(0)}_{t+r_1}-\boldm^{(0)}_{t+r_1-1})(x_{t}+x_{t+1}) \\
&~&&+\sum^{s_2+s_3}_{t=s_2+1}(\boldm^{(0)}_{t+r_2-1}-\boldm^{(0)}_{t+r_2-2})y_t-\boldm^{(0)}_{r_2+s_2+s_3-1}\\
&\le&&\;\boldm^{(0)}_{r_1} + \sum^{s_1+1}_{t=2}(\boldm^{(0)}_{t+r_1-1}-\boldm^{(0)}_{t+r_1-2})  +\sum^{s_1+s_2-1}_{t=s_1+1}(\boldm^{(0)}_{t+r_1}-\boldm^{(0)}_{t+r_1-1}) \\
&~&&+\sum^{s_2+s_3}_{t=s_2+1}(\boldm^{(0)}_{t+r_2-1}-\boldm^{(0)}_{t+r_2-2})-\boldm^{(0)}_{r_2+s_2+s_3-1}=0,
%\boldm^{(0)}_{r_1} + \sum^{s_1+1}_{t=2}(\boldm^{(0)}_{t+r_1-1}-\boldm^{(0)}_{t+r_1-2}) \\
%&~&&+ \sum^{s_1+s_2-1}_{t=s_1+1}(\boldm^{(0)}_{t+r_1}-\boldm^{(0)}_{t+r_1-1}) \\
%&~&&+ \sum^{s_2+s_3}_{t=s_2+1}(\boldm^{(0)}_{t+r_2-1}-\boldm^{(0)}_{t+r_1-2})-\boldm^{(0)}_{r_2+s_2+s_3-1}\\
%&=&&\;0,
\end{alignat*}
where equality equality holds when
\begin{align*}
    &x_t=1\mbox{ for }t\in\{1,\ldots,s_1+1\},\\
    &x_{t}+x_{t+1}=1\mbox{ for } t\in\{s_1+1,\ldots,s_1+s_2-1\},\mbox{ and}\\
    &y_t=1\mbox{ for }t\in\{s_2+1,\ldots,s_2+s_3\},
\end{align*}
and hence~\eqref{equation:casecXY} holds. On the other hand, when~$y_1=1$, let
\begin{align*}
S_1&=\{t:x_{\max\{t-r_1+1,1\}}=1,1\le t\le s_1+r_1\},\\
S_2&=\{t:x_{t-r_1}+x_{t-r_1+1}=0,r_2+1\le t\le r_2+s_2-1\},\\
%\end{align*}
%and
%\begin{align*}
S_3&=\{t:y_{t-r_2+1}=0,r_2+s_2\le t\le r_2+s_2+s_3-1\},
\end{align*}
and notice that
\begin{alignat}{3}\label{case311}
g_{\boldm^{(0)}}(&r_1&&,\boldx) + g_{\boldm^{(0)}}(r_2,\boldy)\nonumber\\
&=&&\; \boldm^{(0)}_{r_1}x_1 + \sum^{s_1+1}_{t=2}(\boldm^{(0)}_{t+r_1-1}-\boldm^{(0)}_{t+r_1-2})x_t+ \nonumber\\
&~&&\boldm^{(0)}_{s_1+r_1}+ \sum^{s_1+s_2-1}_{t=s_1+1}(\boldm^{(0)}_{t+r_1}-\boldm^{(0)}_{t+r_1-1})(x_{t}+x_{t+1}) +\nonumber\\ &~&&\sum^{s_2+s_3}_{t=s_2+1}(\boldm^{(0)}_{t+r_2-1}-\boldm^{(0)}_{t+r_1-2})y_t-\boldm^{(0)}_{r_2+s_2+s_3-1}\nonumber\\
&=&&\;\sum_{t\in S_1}(\boldm^{(0)}_t-\boldm^{(0)}_{t-1})-\sum_{t\in S_2}(\boldm^{(0)}_t-\boldm^{(0)}_{t-1})-\sum_{t\in S_3}(\boldm^{(0)}_t-\boldm^{(0)}_{t-1})\nonumber\\
&=&&\sum_{t\in S_1}1-\sum_{t\in S_2}1-\sum_{t\in S_3}1
\end{alignat}
% \begin{alignat}{3}
% &=&&\; \sum_{i\in S_1}1-\sum^{r_2+s_2+s_3-1}_{t=r_2+1}(\boldm^{(0)}_t-\boldm^{(0)}_{t-1})+\nonumber\\
% &~&&\sum^{s_1+s_2-1}_{t=s_1+1}(\boldm^{(0)}_{t+r_1}-\boldm^{(0)}_{t+r_1-1})(x_{t}+x_{t+1}) +\nonumber\\ &~&&\sum^{s_2+s_3-1}_{t=s_2}(\boldm^{(0)}_{t+r_2}-\boldm^{(0)}_{t+r_1-1})y_t\nonumber\\
% &=&&\; \sum_{i\in S_1}(\boldm^{(0)}_i-\boldm^{(0)}_{i-1})-\nonumber\\
% &~&&\sum^{r_2+s_2-1}_{r_2+1}(\boldm^{(0)}_t-\boldm^{(0)}_{t-1})(1-x_{t-r_1}-x_{t-r_1+1})-\nonumber\\
% &~&&\sum^{r_2+s_2+s_3-1}_{r_2+s_2}(\boldm^{(0)}_t-\boldm^{(0)}_{t-1})(1-y_t)\nonumber\\
% &=&&\;\sum_{i\in S_1}(\boldm^{(0)}_i-\boldm^{(0)}_{i-1})-\sum_{i\in S_2}(\boldm^{(0)}_i-\boldm^{(0)}_{i-1})-\sum_{i\in S_3}(\boldm^{(0)}_i-\boldm^{(0)}_{i-1})\nonumber\\
% &=&&\;\sum_{i\in S_1}1-\sum_{i\in S_2}1-\sum_{i\in S_3}1=0.
% \end{alignat}
Similarly, we have
\begin{align}\label{case312}
g_{\boldm^{(1)}}(r_1,\boldx) + g_{\boldm^{(1)}}(r_2,\boldy)=
%=&\;\sum_{i\in S_1}(\boldm^{(1)}_i-\boldm^{(1)}_{i-1})-\sum_{i\in S_2}(\boldm^{(1)}_i-\boldm^{(1)}_{i-1})-\nonumber\\
%&\;\sum_{i\in S_3}(\boldm^{(1)}_i-\boldm^{(1)}_{i-1})\\
\sum_{t\in S_1}t-\sum_{t\in S_2}t-\sum_{t\in S_3}t.
\end{align}
Equations~\eqref{equation:lemmag+},~\eqref{case311}, and~\eqref{case312} imply that the cardinality of~$S_1$ equals the sum of cardinalities of~$S_2$ and~$S_3$, and in addition, the sum of elements of~$S_1$ equals the sum of elements of~$S_2$ and~$S_3$. Note that the minimum element of~$S_2\cup S_3$ is larger than the maximum element of~$S_1$. This is impossible, unless~$S_1,S_2$, and~$S_3$ are empty, which implies that $x_t=0$~ for~$t\in\{1,\ldots,s_1+1\}$, $x_{t}+x_{t+1}=1$~ for~$t\in\{s_1+1,\ldots,s_1+s_2-1\}$, and $y_t=1$~ for~$t\in\{s_2+1,\ldots,s_2+s_3\}$, and hence~\eqref{equation:casecXY} holds.

\textbf{Case~$(4)$.} $(x_{s_1+s_2+1},y_{s_2+s_3+1})=(1,0)$\hfill

On the one hand, for~$y_1=0$, % arguments similar to the above yield
%\begin{align*}
%   g_{\boldm^{(0)}}(r_1,\boldx) + g_{\boldm^{(0)}}(r_2,\boldy)=-\sum_{i\in S_1}1-\sum_{i\in S_2}1+\sum_{i\in S_3}1=0
%\end{align*}
let
\begin{align*}
S_1&=\{t:x_{\max\{t-r_1+1,1\}}=0,1\le t\le s_1+r_1\},\\
S_2&=\{t:x_{t-r_1}+x_{t-r_1+1}=0,r_2+1\le t\le r_2+s_2-1\},\\
S_3&=\{t:y_{t-r_2+1}=1,r_2+s_2\le t\le r_2+s_2+s_3-1\}.
\end{align*}
We have
\begin{align*}
&g_{\boldm^{(0)}}(r_1,\boldx) + g_{\boldm^{(0)}}(r_2,\boldy) \nonumber\\
=&\; \boldm^{(0)}_{r_1}x_1 + \sum^{s_1+1}_{t=2}(\boldm^{(0)}_{t+r_1-1}-\boldm^{(0)}_{t+r_1-2})x_t
+ \sum^{s_1+s_2-1}_{t=s_1+1}(\boldm^{(0)}_{t+r_1}-\boldm^{(0)}_{t+r_1-1})(x_{t}+x_{t+1})-\nonumber\\
&\;\boldm^{(0)}_{r_1+s_1+s_2-1} + \sum^{s_2+s_3}_{t=s_2+1}(\boldm^{(0)}_{t+r_2-1}-\boldm^{(0)}_{t+r_1-2})y_t\nonumber\\
=& -\boldm^{(0)}_{r_1}(1-x_1)-\sum^{s_1+1}_{t=2}(\boldm^{(0)}_{t+r_1-1}-\boldm^{(0)}_{t+r_1-2})(1-x_t)-\nonumber\\
&\sum^{s_1+s_2-1}_{t=s_1+1}(\boldm^{(0)}_{t+r_1}-\boldm^{(0)}_{t+r_1-1})(1-x_{t}-x_{t+1})+ \sum^{s_2+s_3}_{t=s_2+1}(\boldm^{(0)}_{t+r_2-1}-\boldm^{(0)}_{t+r_1-2})y_t\nonumber\\
=&-\sum_{t\in S_1}(\boldm^{(0)}_t-\boldm^{(0)}_{t-1})-\sum_{t\in S_2}(\boldm^{(0)}_t-\boldm^{(0)}_{t-1})+\sum_{t\in S_3}(\boldm^{(0)}_t-\boldm^{(0)}_{t-1})\nonumber\\
=&-\sum_{t\in S_1}1-\sum_{t\in S_2}1+\sum_{t\in S_3}1=0.
\end{align*}
Then similar to the previous case, we obtain sets with identical cardinalities and sum of elements, and yet the smallest element in one is greater than the largest element in the others. Therefore, it follows that~$S_1,S_2$, and~$S_3$ are empty. Then we have $x_t=1$ for~$t\in\{1,\ldots,s_1+1\}$, $x_{t}+x_{t+1}=1$ for~$t\in\{s_1+1,\ldots,s_1+s_2-1\}$, and $y_t=0$ for~$ t\in\{s_2+1,\ldots,s_2+s_3\}$, and hence~\eqref{equation:casecXY} holds.

On the other hand, for~$y_1=1$, let
%\begin{align*}
%    g_{\boldm^{(0)}}&(r_1,\boldx) + g_{\boldm^{(0)}}(r_2,\boldy)=\sum_{i\in S_1}1-\sum_{j\in S_2}1+\sum_{k\in S_3}1,\mbox{ for}\\
%\end{align*}
%for
\begin{align*}
S_1&=\{t:x_{\max\{t-r_1+1,1\}}=1,1\le t\le s_1+r_1\},\\
S_2&=\{t:x_{t-r_1}+x_{t-r_1+1}=0,r_2+1\le t\le r_2+s_2-1\},\\
S_3&=\{t:y_{t-s_2+1}=1,r_2+s_2\le t\le r_2+s_2+s_3-1\}.
\end{align*}
We have
\begin{align}\label{equation:vandermondecase1}
&g_{\boldm^{(0)}}(r_1,\boldx) + g_{\boldm^{(0)}}(r_2,\boldy) \nonumber\\
=& \;\boldm^{(0)}_{r_1}x_1 + \sum^{s_1+1}_{t=2}(\boldm^{(0)}_{t+r_1-1}-\boldm^{(0)}_{t+r_1-2})x_t +\boldm^{(0)}_{s_1+r_1}+ \sum^{s_1+s_2-1}_{t=s_1+1}(\boldm^{(0)}_{t+r_1}-\boldm^{(0)}_{t+r_1-1})(x_{t}+x_{t+1})-\nonumber\\
&\;\boldm^{(0)}_{r_1+s_1+s_2-1} + \sum^{s_2+s_3}_{t=s_2+1}(\boldm^{(0)}_{t+r_2-1}-\boldm^{(0)}_{t+r_1-2})y_t\nonumber\\
=&\; \boldm^{(0)}_{r_1}x_1+\sum^{s_1+1}_{t=2}(\boldm^{(0)}_{t+r_1-1}-\boldm^{(0)}_{t+r_1-2})x_t-\nonumber\\
&\sum^{s_1+s_2-1}_{t=s_1+1}(\boldm^{(0)}_{t+r_1}-\boldm^{(0)}_{t+r_1-1})(1-x_{t}-x_{t+1}) + \sum^{s_2+s_3}_{t=s_2+1}(\boldm^{(0)}_{t+r_2-1}-\boldm^{(0)}_{t+r_1-2})y_t\nonumber\\
=&\sum_{t\in S_1}(\boldm^{(0)}_t-\boldm^{(0)}_{t-1})-\sum_{t\in S_2}(\boldm^{(0)}_t-\boldm^{(0)}_{t-1})+\sum_{t\in S_3}(\boldm^{(0)}_t-\boldm^{(0)}_{t-1})\nonumber\\
=&\sum_{t\in S_1}1-\sum_{t\in S_2}1+\sum_{t\in S_3}1=0.
\end{align}
Similarly, we have
\begin{align}\label{equation:vandermondecase2}
g_{\boldm^{(1)}}(r_1,\boldx) + g_{\boldm^{(1)}}(r_2,\boldy)&=
\sum_{t\in S_1}t-\sum_{t\in S_2}t+\sum_{t\in S_3}t\nonumber\\
g_{\boldm^{(2)}}(r_1,\boldx) + g_{\boldm^{(2)}}(r_2,\boldy)&=\sum_{t\in S_1}t^2-\sum_{t\in S_2}t^2+\sum_{t\in S_3}t^2
% &\sum_{i\in S_1}(\boldm^{(2)}_i-\boldm^{(2)}_{i-1})-\nonumber\\
% &\sum_{i\in S_2}(\boldm^{(2)}_i-\boldm^{(2)}_{i-1})+\nonumber\\
% &\sum_{i\in S_3}(\boldm^{(2)}_i-\boldm^{(2)}_{i-1})\nonumber\\
% =&
% \sum_{i\in S_1}i^2-\sum_{i\in S_2}i^2+\sum_{i\in S_3}i^2=0
\end{align}
According to~\eqref{equation:vandermondecase1} and~\eqref{equation:vandermondecase2}, the following
linear equation
\begin{align}\label{linear}
A\boldsymbol{x}=\begin{bmatrix}
\sum_{t\in S_1}1 &\sum_{t\in S_2}1 & \sum_{t\in S_3}1\\
\sum_{t\in S_1}t &\sum_{t\in S_2}t & \sum_{t\in S_3}t\\
\sum_{t\in S_1}t^2 &\sum_{t\in S_2}t^2 & \sum_{t\in S_3}t^2
\end{bmatrix}
\begin{bmatrix}
x_1\\
x_2\\
x_3\\
\end{bmatrix}
=0
\end{align}
has a nonzero solution~$(x_1,x_2,x_3)=(1,-1,1)^\top$. However, according to the linearity of the determinant, sharethe determinant
\begin{align}
\det(A)=&\sum_{i\in S_1,j\in S_2,k\in S_3}\det \begin{pmatrix} 1&1&1\\ i&j&k\\ i^2&j^2&k^2 \end{pmatrix}\nonumber\\
=&\sum_{i\in S_1,j\in S_2,k\in S_3} (j-i)(k-i)(k-j)
\end{align}
is strictly positive since $\max_{i\in S_1}i<\min_{j\in S_2}j<\min_{k\in S_3}k$. Thus, Eq.~\eqref{linear} has no nonzero solution unless $A=0$, which implies that $S_1,S_2$, and $S_3$ are empty. Therefore,  $x_t=0$ for~$t\in\{1,\ldots,s_1+1\}$, $x_{t}+x_{t+1}=1$ for~$t\in\{s_1+1,\ldots,s_1+s_2-1\}$, and $y_t=0$ for~$ t\in\{s_2+1,\ldots,s_2+s_3\}$, which implies~\eqref{equation:casecXY}.
\end{proof}
\section{Encoding and Decoding Algorithms}\label{section:decoding}
We now show how to use Theorem~\ref{theorem:main} to construct an encoding algorithm and a decoding algorithm. %encoding/decoding algorithm.
Similar to the two layer encoding method described in \cite{brakensiek2016efficient}, we use the~$f(\boldc)$ and~$h(\boldc)$ redundancies~\eqref{equation:fgfunction} to protect the sequence~$\boldc$ from two deletions in the first layer. In the second layer, the~$f(\boldc)$ and~$h(\boldc)$ redundancies are protected again by their corresponding~$f(f(\boldc),h(\boldc))$ and~$h(f(\boldc),h(\boldc))$ redundancies. Since~$f(f(\boldc),h(\boldc))$ and~$h(f(\boldc),h(\boldc))$ are short, they can be protected by an inefficient 3-fold repetition code.
%Then we use a relatively inefficient but short  to protect the~$f$ and~$g$ redundancies from two deletions.
Specifically, for any sequence~$\boldc\in\{0,1\}^n$, the encoding function is
%let the encoding function
\begin{align}\label{equation:encoding}
\mathcal{E}(\boldc)=(\boldc,f(\boldc),h(\boldc),r_3(f(f(\boldc),h(\boldc))),r_3(h(f(\boldc),h(\boldc)))),
\end{align}
where~$r_3$ is a 3-fold repetition encoding function. The length of the first layer redundancy~$f(\boldc),h(\boldc)$ is~$N_1=7\log n +2$. The length of the 3-fold repetition of the second layer redundancy~$r_3(f(f(\boldc),h(\boldc))),r_3(h(f(\boldc),h(\boldc)))$ is~$N_2=21\log(7\log n+2) +6$.
%Specifically, the encoding function appends the parity checks~$f(\boldc)$ and~$h(\boldc)$ to the end of~$\boldc$ and then appends the three times repetition of the parity redundancy of the parity checks~$f(f(\boldc),h(\boldc))$ and~$h(f(\boldc),h(\boldc))$ to the end.
The length of the codeword~$\mathcal{E}(\boldc)$ is
\begin{align*}
N=n+N_1+N_2=n+7\log n+2 +21\log(7\log n+2) +6=n+7\log n +o(\log n).
\end{align*}
Clearly, the
computation of the function~$\mathcal{E}(\boldc)$ can be done in linear time.

%Toward decoding successfully
To conveniently describe the decoding algorithm, two building blocks are needed. The first is a 3-fold repetition decoding function
\begin{align*}
\mathcal{D}_1:\{0,1\}^{3N_2-2}\rightarrow\{0,1\}^{N_2}
\end{align*}
that takes a subsequence~$\boldd_1\in\{0,1\}^{3N_2-2}$ of a 3-fold repetition codeword~$r_3(\bolds_1)\in\{0,1\}^{3N_2}$ for some~$\bolds_1\in\{0,1\}^{N_2}$ as input, and outputs an estimate~$\tilde{\bolds}_1$ of the sequence~$\bolds_1$. The second is a decoding function
\begin{align*}
\mathcal{D}_2:\{0,1\}^{n-2}\times \{0,1\}^{7\log n+2}\rightarrow\{0,1\}^n
\end{align*}
that takes a subsequence~$\boldd_2\in\{0,1\}^{n-2}$ of some~$\bolds_2\in\{0,1\}^{n}$, redundancy~$f(\bolds_2)$, and redundancy~$h(\bolds_2)$ as input, and outputs an estimate~$\tilde{\bolds}_2$ of the sequence~$\bolds_2$. The 3-fold repetition decoding~$\mathcal{D}_1$ can be implemented by adding two bits to~$\boldd_1$ such that the length of each run is a multiple of~$3$, which can obviously be done in linear time.
%This can be done in linear time.
According to Theorem~\ref{theorem:main}, there exists a decoding function~$\mathcal{D}_2$ that recovers the original sequence correctly given its~$f$ and~$h$ redundancy. The linear complexity of~$\mathcal{D}_2$ will be shown later in this section.

%Now suppose that we have decoding functions~$\mathcal{D}_1$ and~$\mathcal{D}_2$ that return the correct estimates.
The functions~$\mathcal{D}_1$ and~$\mathcal{D}_2$ are used as subroutines to describe the decoding procedure that is given in Algorithm~\ref{algorithm:meta}.
%Now we use
First, we use the function~$\mathcal{D}_1$ to recover the second layer redundancy~$f(f(\boldc),h(\boldc))$ and~$h(f(\boldc),h(\boldc))$ from the 3-fold repetition code. Then, by applying~$\mathcal{D}_2$ and using the second layer redundancy~$f(f(\boldc),h(\boldc))$ and~$h(f(\boldc),h(\boldc))$,
the first layer redundancy~$f(\boldc)$ and~$h(\boldc)$ can be recovered. Finally and similarly, the first layer redundancy~$f(\boldc)$ and~$h(\boldc)$ can be used to recover the original sequence~$\boldc$, with the help of~$\mathcal{D}_2$. In the case of single deletion, 
Algorithm~\ref{algorithm:meta} outputs the orginal sequence~$\boldc$. One can also use a VT decoder (see \cite{levenshtein1966binary}), which has a simper implementation and~$O(n)$ time complexity.

%The decoding procedure
% can be described as a meta-algorithm (Algorithm~\ref{algorithm:meta}) that uses~$\mathcal{D}_1$ and~$\mathcal{D}_2$ to recover~$\boldc$. The input of Algorithm  is a length~$N-2$ subsequence~$\boldd$ of~$\mathcal{E}(\boldc)$. The problem needed to be addressed in Algorithm~\ref{algorithm:meta} is to identify which part is a subsequence of~$\boldc$,~$f(\boldc),h(\boldc)$, or~$r_3(f(f(\boldc),h(\boldc))),r_3(h(f(\boldc),h(\boldc)))$ respectively, and thus to specify the input of decoding functions~$\mathcal{D}_1$ and~$\mathcal{D}_2$.
%\NetanelComment{Consider adding stopping conditions to this algorithm. That is, if the first call to~$\cD_1$ detects two deletions, you're done.}
\begin{algorithm}[h]\label{algorithm:meta}
\KwIn{Subsequence~$\boldd\in\{0,1\}^{N-2}$ of~$\mathcal{E}(\boldc)$}
\KwOut{The sequence $\boldc$.}
%\textbf{} \;
\textit{layer2\_redundancy}~$=\mathcal{D}_1(\boldd^{(N-N_2+1,N-2)})$\;
	\eIf{\textup{detect two deletions after~the first run in~$\boldd_{N-N_2+1,N-2}$} }
      {
      return~$\boldd^{(1,n)}$\;
     }{
%\NetanelComment{Something here is strange, the start index is larger than the end index. Perhaps the end index should be~$N-2$? Even if so, the length of the input here is~$N_2-4$, that does not work out with the definition of~$\cD_1$.}\;
$L\triangleq$~the length of the longest suffix of~$\boldd$ that is a subsequence of~$r_3(\mbox{\textit{layer2\_redundancy})}$\;
\textit{layer1\_redundancy}~$=\mathcal{D}_2(\boldd^{(N-N_1+1-L,N-2-L)}, \mbox{\textit{layer2\_redundancy}})$\;
%\NetanelComment{Same here. The length of the input is~$N_1-3$. Perhaps you want to define~$\cD_2$ on input of varying length?}\;
$\boldc = \mathcal{D}_2(\boldd^{(1,n-2)},\mbox{\textit{layer1\_redundancy}})$\;
    \caption{{\bf Decoding} %\label{Algorithm}
    }
\Return{$\boldc$.}
}
\end{algorithm}
%\NetanelComment{Change to Theorem.}
\begin{theorem}
If the functions~$\mathcal{D}_1$ and~$\mathcal{D}_2$ provide the correct estimates in~$O(n)$ time, then given a~$N-2$ subsequence of~$\mathcal{E}(\boldc)$, Algorithm~\ref{algorithm:meta} returns the original sequence~$\boldc$ in~$O(n)$ time.

%\NetanelRemove{If the functions~$\mathcal{D}_1$ and~$\mathcal{D}_2$ provide the correct estimates and have linear time complexity, then Algorithm~\ref{algorithm:meta} returns the original sequence~$\boldc$ given a length~$N-2$ subsequence of~$\mathcal{E}(\boldc)$. The complexity of Algorithm~\ref{algorithm:meta} is~$O(n)$.}

\end{theorem}
\begin{proof}
To prove the correctness of Algorithm~\ref{algorithm:meta}, it suffices to show~the following
\begin{itemize}
\item[$(1)$.]\label{case:1} $\boldd^{(N-N_2+1,N-2)}$ is a length~$N_2-2$ subsequence of the repetition code~$r_3(f(f(\boldc),h(\boldc))),r_3(h(f(\boldc),h(\boldc))))$.
\item[$(2)$.]\label{case:2} $\boldd^{(N-N_1+1-L,N-2-L)}$ is a length~$N_1-2$
%\NetanelComment{$N_1-3$?. }
subsequence of the~$f(\boldc),h(\boldc)$ redundancy.
\item[$(3)$.]\label{case:3} $\boldd^{(1,n-2)}$ is a length~$n-2$ subsequence of the sequence~$\boldc$.
\end{itemize}
Since~$\boldd$ is a length~$N-2$ subsequence of~$\mathcal{E}(\boldc)$,~$d_{n-2}$ must be either the~$(n-2)$-th, the~$(n-1)$-th or the $n$-th bits of~$\mathcal{E}(\boldc)$, and hence~$(3)$ must hold. Similarly,~$(1)$ holds by looking at~$\boldd$ and~$\mathcal{E}(\boldc)$ in reverse order.
%because the last~$N_2-2$ bits of~$\boldd$ length of the repetition code~$r_3(f(f(\boldc),h(\boldc))),r_3(h(f(\boldc),h(\boldc))))$ is~$N_2$.
By the~definition of~$L$,~$d_{N-2-L}$ is~the $i_1$-th bit of~$\mathcal{E}(\boldc)$ for some~$i_1\le n+N_1$.
Since~$(1)$ holds, we have that~$L$ is either the~$N_2$-th, the $(N_2-1)$-th, or the~$(N_2-2)$-th bits of~$\cE(\boldc)$.
Therefore,~$d_{N-N_1+1-L}$ is~the~$i_2$-th bit of~$\mathcal{E}(\boldc)$ for some~$i_2\ge N-N_1+1-L>n$. Since~$(f(\boldc),h(\boldc))=\mathcal{E}(\boldc)^{(n+1,n+N_1)}$,~$(2)$ must hold.

Since finding~$L$ has~$O(N_2)$ complexity, the complexity of Algorithm~\ref{algorithm:meta} is~$O(N)=O(n)$, given that the complexities of the functions ~$\mathcal{D}_1$ and~$\mathcal{D}_2$ are linear.
\end{proof}
%After receiving a subsequence~$\boldd\in\{0,1\}^{N-2}$, first recover parity checks~$f(f(\boldc),h(\boldc)),h(f(\boldc),h(\boldc)))$ from the last $21\log\log n-2$ bits of~$\boldd$. Then find the longest suffix~$\boldd'$ of~$\boldd$ that is a subsequence of~$r_3(f(f(\boldc),h(\boldc))),r_3(h(f(\boldc),h(\boldc))))$. Then look at the last~$7\log n-2$ bits to the left of~$\boldd'$. These bits must be a subsequence of $f(\boldc),h(\boldc)$. Then we can recover~$f(\boldc),h(\boldc)$ using the parity checks~$f(f(\boldc),h(\boldc)),h(f(\boldc),h(\boldc)))$. Finally, we can use parity checks~$f(\boldc),h(\boldc)$ to recover~$\boldc$ from the first~$n-2$ bits of~$\boldd$, which must be a subsequence of~$\boldc$.
%The remaining problem is to find an algorithm that implements~$\mathcal{D}_2$ with linear complexity.
We are left to implement~$\cD_2$ with linear complexity.
%Specifically,
In particular, we need to
recover the sequence $\boldc\in\{0,1\}^n$ from its length~$n-2$ subsequence~$\boldd$ in time~$O(n)$, given the redundancy~$f(\boldc)$ and~$h(\boldc)$.
Note that there are~$O(n^2)$ supersuquences of~$\boldd$ of length~$n$, and~$f$ and~$h$ can be computed on each of them in~$O(n)$. Hence, the brute force approach would require~$O(n^3)$.
%using }brute force
%\NetanelAddition{enumeration of }
%all possible length~$n$ supersequences of~$\boldd$ and
%%check
%\NetanelAddition{checking }
%their redundancies, the time complexity \NetanelAddition{is }~$O(n^3)$, \NetanelAddition{since }there are~$O(n^2)$ such supersequences, and each of~\NetanelAddition{which requires}~$O(n)$ \NetanelAddition{time }to compute~$f$ and~$h$\NetanelRemove{redundancy}.

To achieve linear time complexity, we first recover~$\1_{10}(\boldc)$, which is an~$(n-3)$-subsequence of~$\1_{10}(\boldc)\in\{0,1\}^{n-1}$, and then use it to recover~$\boldc$. In particular, we find the positions of the deleted bits by an iterative updating algorithm, rather than by exhaustive search, and hence linear complexity is obtained.
%based on the~$10$-indicator~$\1_{10}(\boldc)$. Recall from Lemma~\ref{lemma:twoDeletions10} \NetanelAddition{that}~$\1_{10}(\boldd)$ is a length~$n-3$ subsequence of~$\1_{10}(\boldc)\in\{0,1\}^{n-1}$, the indicator vector~$\1_{10}(\boldc)$ can be recovered using an updating algorithm. Specifically, we search for all possible two insertions to~$\1_{10}(\boldd)$ such that the resulting supersequence satisfies the parity check~$f_1(\boldc)$. At each search step, the parity checks of the resulting supersequence are updated to see if they satisfy~$f(\boldc)$.
Furthermore, the uniqueness of the obtained sequence is guaranteed by Lemma~\ref{lemma:10indicator}.
%, we have that~$\1_{10}(\boldc)$ is the unique supersequence of~$\1_{10}(\boldd)$ that satisfies~$f(\boldc)$.}

After recovering~$\1_{10}(\boldc)$,
We can find all
length~$n$ supersequences~$\boldc'$ of~$\boldd$ such that~$\1_{10}(\boldc')=\1_{10}(\boldc)$. It is shown that there are at most~$4$ such possible supersequences, %\NetanelComment{Doesn't this suggest that the~$h$ redundancy can be replaced with 2 bits?}
 and since~Theorem~\ref{theorem:main} guarantees uniqueness, the right~$\boldc$ is found by computing and comparing~$h$.
%Then we can check which of the~$4$ supersequences~$\boldc'$ satisfies~$h(\boldc)$. According to Theorem~\ref{theorem:main}, there is a unique such supersequence.

\subsection{Recovering~$\1_{10}(\boldc)$}
For~$1\le i\le 2n-2$, let
\begin{align}
	&p_i\triangleq\begin{cases}\label{equation:pipj}
	n-i & \mbox{if }1\le i\le n-1\\
	i-n+1 &\mbox{if }n\le i\le 2n-2
	\end{cases}\mbox{, and}\\
    &b_i\triangleq\begin{cases}
	1 & \mbox{if }1\le i\le n-1\\
	0 &\mbox{if }n\le i\le 2n-2
	\end{cases}.
\end{align}
Given a subsequence~$\boldd\in\{0,1\}^{n-2}$ of~$\boldc$, 
%The above multi-case definition of~$\boldd(i,j)$ is nearly impossible to comprehend. The sentence you wrote after it gives a much more concise description of what's going on. Can't we use it as a definition, and give one or two of the six above cases as an example? For instance:
let~$\1_{10}(\boldd)=(r_1,\ldots,r_{n-3})$, and let~$\boldd:[2n-2]\times[2n-2]\to \{0,1\}^n\cup\{\star\}$ be defined as
$$ \boldd(i,j) =
\begin{cases}
	(r_1,r_2,\ldots,r_{p_i-1},b_i,r_{p_i},\ldots,r_{p_j-2},b_j,r_{p_j-1},\ldots,r_{n-3}) & \mbox{if }p_i<p_j\\
	(r_1,r_2,\ldots,r_{p_j-1},b_j,r_{p_j},\ldots,r_{p_i-2},b_i,r_{p_i-1},\ldots,r_{n-3}) & \mbox{if }p_i>p_j\\
	\star & \mbox{if }p_i=p_j
\end{cases}, $$
that is, $\boldd(i,j)$ results from~$\1_{10}(\boldd)$ inserting~$b_i$ at position~$p_i$ and~$b_j$ in position~$p_j$ of~$\1_{10}(\boldd)$, if~$p_i\ne p_j$. Notice that~$\boldd(i,j)$ is one possible way of correcting two deletions in the sequence~$\1_{10}(\boldd)$. 
For~$e\in\{0,1,2\}$ define matrices~$\{A^{(e)}\}_{e=0}^2$ as follows.% \NetanelComment{Use~$(e)$ instead of~$e$ to avoid confusion with power.}
\begin{align*}
A^{(e)}_{i,j}=\begin{cases}
			\boldd(i,j)\cdot \boldm^{(e)}-\sum^{n-3}_{i=1}\boldm^{(e)}_i\1_{10}(\boldd)_i  &\mbox{if $\boldd(i,j)\ne \star$.}\\
            %= b(O(i))\boldm^{(e)}_{p(O(i))}+b(O(j))\boldm^{(e)}_{p(O(j))}+&\\
            %\sum^{n-3}_{i=1}\1_{10}(\boldd)_i[i^e\1_{\min\{p(O(i)),p(O(j))\}<i+1}&\\
            %+(i+1)^e\1_{\max\{p(O(i)),p(O(j))\}<i+2}] & \mbox{if $\boldd(i,j)$ is legitimate.}\\
			\star &\mbox{if $\boldd(i,j)=\star$.}
		\end{cases}.
\end{align*}
Notice that~$A^{(e)}_{i,j}$ is the difference in entry~$e$ of the~$f$ redundancies of~$\boldd(i,j)$ and~$\1_{10}(\boldd)$, i.e., $A_{i,j}^{(e)}=f(\boldd(i,j))_e-f(\1_{10}(\boldd))_e$.

%\NetanelRemove{to be the increase from index weighted sum of bits of~$\1_{10}(\boldd)$ to that of~$\boldd(i,j)$. }

We~prove the following properties of~$A^{(e)}$. In the first property, we give an explicit expression for the matrices~$A_{i,j}^{(e)}$ in terms of~$\1_{10}(\boldd)$, $p_i$,~$p_j$,~$b_i$, and~$b_j$. The expression will be used for calculating~$A_{i,j}^{(e)}$ in constant time from its neighboring entries during~$\mathcal{D}_2$. In the following we use~$\delta (x)$ to denote the indicator of the event~$x$, where~$\delta (x)=1$ if and only if~$x$ is true. %\NetanelComment{Give the reader a hint here for what the following proposition means, and what will it be used for. Also, I assume that~$\1_{x}$ below is a function which gets~$1$ if~$x$ is true and~$0$ otherwise. If so, I don't think that this a very good notation, since~$\1$ is used extensively for other purposes. Perhaps consider using~$\delta(x)$ (and introduce it explicitly to the reader before using it)?}
%\NetanelRemove{
% \begin{proposition}\label{proposition:matrixAF_}
% Let
% \begin{align}\label{equation:increasematrix_}
% F^e_{i,j}&=b_i\boldm^{(e)}_{p_i}+b_j\boldm^{(e)}_{p_j}+
%             \sum^{n-3}_{k=1}\1_{10}(\boldd)_k[(k+1)^e\1_{\min\{p_i,p_j\}<k+1}   +(k+2)^e\1_{\max\{p_i,p_j\}<k+2}]
% \end{align}
% Then
% \begin{align*}
% A^{(e)}_{i,j}=\begin{cases}
% 			F^e_{i,j}  &\mbox{if $\boldd(i,j)$ is legitimate.}\\
%             %= b(O(i))\boldm^{(e)}_{p(O(i))}+b(O(j))\boldm^{(e)}_{p(O(j))}+&\\
%             %\sum^{n-3}_{i=1}\1_{10}(\boldd)_i[i^e\1_{\min\{p(O(i)),p(O(j))\}<i+1}&\\
%             %+(i+1)^e\1_{\max\{p(O(i)),p(O(j))\}<i+2}] & \mbox{if $\boldd(i,j)$ is legitimate.}\\
% 			* &\mbox{if $\boldd(i,j)$ is illegitimate.}
% 		\end{cases}.
% \end{align*}
% \end{proposition}}
\begin{proposition}\label{proposition:matrixAF}
	If~$A_{i,j}^{(e)}\ne \star$ then
	\begin{align}\label{equation:increasematrix}
	A^{(e)}_{i,j}&=b_i\boldm^{(e)}_{p_i}+b_j\boldm^{(e)}_{p_j}+
	\sum^{n-3}_{k=1}\1_{10}(\boldd)_k[(k+1)^e\delta(\min\{p_i,p_j\}<k+1)   +(k+2)^e\delta(\max\{p_i,p_j\}<k+2)].
	\end{align}
\end{proposition}
\begin{proof}
The difference between $\sum^{n-3}_{k=1}\boldm^{(e)}_k\1_{10}(\boldd)_k$ and~$\boldd(i,j)\cdot \boldm^{(e)}$~consists of two parts. The first part follows~from the two inserted~bits, and can be written as
\begin{align}\label{equation:increasebits}
b_i\boldm^{(e)}_{p_i}+b_j\boldm^{(e)}_{p_j}
\end{align}
The second part follows from the shift of bits in~$\1_{10}(\boldd)_k$ that is caused by the insertions of two bits~$b_i$ and~$b_j$. Each bit~$\1_{10}(\boldd)_k$ shifts from position~$k$ to position~$k+1$ if one insertion occurs~before~$\1_{10}(\boldd)_k$, i.e.,~$\min\{p_i,p_j\}<k+1$ and~$\max\{p_i,p_j\}\ge k+2$. The resulting difference~is given by
\begin{align}\label{equation:increaseoneshift}
&\sum^{n-3}_{k=1}\1_{10}(\boldd)_k\delta (\min\{p_i,p_j\}<k+1)\delta(\max\{p_i,p_j\}\ge k+2)(\boldm^{(e)}_{k+1}-\boldm^{(e)}_k)\nonumber\\=&\sum^{n-3}_{k=1}\1_{10}(\boldd)_k\delta(\min\{p_i,p_j\}<k+1)\delta(\max\{p_i,p_j\}\ge k+2)(k+1)^e.
\end{align}
The bit~$\1_{10}(\boldd)_k$ shifts from position~$k$ to~$k+2$ if two insertions occur~before~$\1_{10}(\boldd)_k$, i.e.,~$\max\{p_i,p_j\}<k+2$. The corresponding difference~is given by
\begin{align}\label{equation:increasetwoshifts}
&\sum^{n-3}_{k=1}\1_{10}(\boldd)_k\delta(\min\{p_i,p_j\}<k+1)\delta(\max\{p_i,p_j\}< k+2)\1_{10}(\boldd)_k(\boldm^{(e)}_{k+2}-\boldm^{(e)}_k)\nonumber\\
=&\sum^{n-3}_{k=1}\1_{10}(\boldd)_k\delta(\max\{p_i,p_j\}< k+2)[(k+1)^e+(k+2)^e].
\end{align}
Combining~\eqref{equation:increaseoneshift} and~\eqref{equation:increasetwoshifts}, we have that the difference that results~from the second part is given by
\begin{align*}
\sum^{n-3}_{k=1} \1_{10}(\boldd)_k[(k+1)^e\delta(\min\{p_i,p_j\}<k+1)+(k+2)^e\delta(\max\{p_i,p_j\}<k+2)],
\end{align*}
that together~with~\eqref{equation:increasebits}, implies~\eqref{equation:increasematrix}.
\end{proof}
The following shows that the entries of each~$A^{(e)}$ are non-decreasing in rows and columns, and that the respective sequences~$\boldd(i,j)$ that lie in the same column or the same row, are unique given each entry value. This property guarantees a simple algorithm for finding a sequence~$\boldd(i,j)$ with a given value~$A^{(e)}_{i,j}$ by decreasing~$i$ or increasing~$j$ by~$1$ in each step.

%\NetanelRemove{\begin{proposition}\label{proposition:nondecreasing_}
% The matrix~$A^{(e)}_{i,j}$ is non-deceasing in each row and each column, i.e., we have~$A^{(e)}_{i_1,j}\le A^{(e)}_{i_2,j}$ and~$A^{(e)}_{i,j_1}\le A^{(e)}_{i,j_2}$ for~$i_1<i_2$,~$j_1<j_2$ and legitimate sequences~$\boldd(i_1,j),\boldd(i_2,j),\boldd(i,j_1)$, and~$\boldd(i,j_2)$. Moreover, if~$A^{(e)}_{i_1,j}= A^{(e)}_{i_2,j}$ or~$A^{(e)}_{i,j_1}= A^{(e)}_{i,j_2}$, then~$\boldd(i_1,j)=\boldd(i_2,j)$ or~$\boldd(i,j_1)=\boldd(i,j_2)$ respectively.
% \end{proposition}}
\begin{proposition}\label{proposition:nondecreasing}
	For every~$i,j$ and~$i_1<i_2$, $j_1<j_2$, if neither of~$\boldd(i_1,j),\boldd(i_2,j),\boldd(i,j_1)$, and~$\boldd(i,j_2)$ equals~$\star$, then $A^{(e)}_{i_1,j}\le A^{(e)}_{i_2,j}$ and~$A^{(e)}_{i,j_1}\le A^{(e)}_{i,j_2}$. Moreover, if~$A^{(e)}_{i_1,j}= A^{(e)}_{i_2,j}$ (resp.~$A^{(e)}_{i,j_1}= A^{(e)}_{i,j_2}$), then~$\boldd(i_1,j)=\boldd(i_2,j)$ (resp.~$\boldd(i,j_1)=\boldd(i,j_2)$).
\end{proposition}

\begin{proof}
By symmetry we only need to prove that the matrix~$A^{(e)}$ is non-decreasing in each column, for which it suffices to prove that:
\begin{enumerate}
	\item[$(1)$.] $A^{(e)}_{i_1,j}\le A^{(e)}_{i_2,j}$ for~$1\le i_1<i_2\le n-1$.
	\item[$(2)$.] $A^{(e)}_{n-1,j}\le A^{(e)}_{n,j}$.
	\item[$(3)$.] $A^{(e)}_{i_1,j}\le A^{(e)}_{i_2,j}$ for~$n\le i_1<i_2\le 2n-2$.
\end{enumerate}

For~$(2)$, the only difference between~$\boldd(n-1,j)$ and~$\boldd(n,j)$ is that their first bits are~$0$ and~$1$ respectively, and hence~$A^{(e)}_{n-1,j}+1= A^{(e)}_{n,j}$. We are left to show~$(1)$ and~$(3)$.

$(1)$: For~$1\le i_1<i_2\le n-1$, we have~$b_{i_1}=b_{i_2}=0$ and~$p_{i_1}>p_{i_2}$.
Let~$\boldd'(i_1,j)\in\{0,1\}^{n-2}$ and~$\boldd'(i_2,j)\in\{0,1\}^{n-2}$~be two subsequences of~$\boldd(i_1,j)$ and~$\boldd(i_2,j)$ respectively after deleting the~$p_j$-th bit from both~$\boldd(i_1,j)$ and~$\boldd(i_2,j)$, and similarly, let~$\boldm^{(e),p_j}=(\boldm^{(e)}_1,\boldm^{(e)}_2,\ldots,\boldm^{(e)}_{p_j-1},\boldm^{(e)}_{p_j+1},$ $\ldots,\boldm^{(e)}_{n-1})$ be a subsequence of~$\boldm^{(e)}$ after deleting the~$p_j$-th entry. Then, according to~\eqref{equation:VTccprime} and~\eqref{equation:differenceofweightedsum}, we have
\begin{align}\label{equation:differencematrixentry}
A^{(e)}_{i_2,j}-A^{(e)}_{i_1,j}&=\boldd(i_2,j)\cdot \boldm^{(e)}-\boldd(i_1,j)\cdot \boldm^{(e)}\nonumber \\
&=\boldd'(i_2,j)\cdot \boldm^{(e),p_j}-\boldd'(i_1,j)\cdot \boldm^{(e),p_j}\nonumber\\
&=g(k_1,\boldd'(i_2,j)^{(k_1,k_2)},\boldd'(i_1,j)_{k_2})\nonumber\\
&\ge  0,
\end{align}
where~$k_1=p_{i_2}-\delta(p_{i_2}>p_j)$ and~$k_2=p_{i_1}-\delta(p_{i_1}>p_j)$ ~are the indices whose deletion~from~$\boldd'(i_2,j)$ and~$\boldd'(i_1,j)$, respectively, results in~$\1_{10}(\boldd)$. Similarly, as in~the proof in Lemma~\ref{lemma:VT}, the last inequality follows from the fact that~$\boldd'(i_1,j)_{k_2}=b_{i_1}=0$. Furthermore, equality~holds when~$\boldd'(i_2,j)^{(k_1,k_2)}=0$ and~$\boldd'(i_1,j)_{k_2}=0$, which implies that~$\boldd'(i_1,j)=\boldd'(i_2,j)$, and hence~$\boldd(i_1,j)=\boldd(i_2,j)$.

$(3)$: For~$n\le i_1<i_2\le 2n-2$, we have~$b_{i_1}=b_{i_2}=1$ and~$p_{i_1}<p_{i_2}$. Similar to~\eqref{equation:differencematrixentry}, we have~that
\begin{align*}
A^{(e)}_{i_1,j}-A^{(e)}_{i_2,j} =g(k_1,\boldd'(i_1,j)^{(k_1,k_2)},\boldd'(i_2,j)_{k_2})\le 0,
\end{align*}
where~$k_1=p_{i_1}-\delta(p_{i_1}>p_j)$ and~$k_2=p_{i_2}-\delta(p_{i_2}>p_j)g$~are the indices whose deletion~from~$\boldd'(i_1,j)$ and~$\boldd'(i_2,j)$, respectively, results in~$\1_{10}(\boldd)$. The last inequality follows from the fact that~$\boldd'(i_2,j)_{k_2}=b_{i_2}=1$, and
~equality holds when~$\boldd(i_1,j)=\boldd(i_2,j)$.
\end{proof}
\begin{remark}\label{remark:entrybound}
From proposition~\ref{proposition:nondecreasing}, we have that%~\NetanelComment{There should be some quantifier on~$i,j$ and~$e$ below.}
\begin{align*}
0 =A^{(e)}_{1,2}\le A^{(e)}_{i,j}\le A^{(e)}_{2n-2,2n-3}\le \boldm^{(e)}_{n-1}+ \boldm^{(e)}_{n-2}\le n_e,~,1\le i,j\le 2n-2,~A^{(e)}_{i,j}\ne \star
\end{align*}
where~$n_0=2n,~n_1=n^2,~n_2=n^3$.
\end{remark}
Our goal is to find a sequence~$\boldd(i,j)\ne \star$ for which
\begin{align}\label{equation:paritymoduloincrease}
&A^{(e)}_{i,j}\equiv f_1(\boldc)-\sum^{n-3}_{i=1}\boldm^{(e)}_i\1_{10}(\boldd)_i \bmod n_e
\end{align}
for every~$e\in\{0,1,2\}$.
In addition, the sequence~$\boldd(i,j)$ cannot contain adjacent~$1$'s, i.e.,
\begin{align}\label{equation:checking11}
&\boldd(i,j)_{p_i-1}\cdot\boldd(i,j)_{p_i}=\boldd(i,j)_{p_i}\cdot\boldd(i,j)_{p_i+1}=0\nonumber\\
&\boldd(i,j)_{p_j-1}\cdot\boldd(i,j)_{p_j}=\boldd(i,j)_{p_j}\cdot\boldd(i,j)_{p_j+1}=0,
\end{align}
and~from Lemma~\ref{lemma:10indicator}, such~$\boldd(i,j)$ equals~$\1_{10}(\boldc)$. Moreover, since~Remark~\ref{remark:entrybound} implies that~$0\le A^{(e)}_{i,j}\le n_e$, it follows that the modular equality in~\eqref{equation:paritymoduloincrease} is unnecessary, i.e.,~it suffices to find a sequence~$\boldd(i,j)\ne\star$ that satisfies~\eqref{equation:checking11} and
%\NetanelComment{Who's~$a_e$?}
\begin{align}\label{equation:parityincrease}
A^{(e)}_{i,j}=a_e\triangleq f_e(\boldc)-\sum^{n-3}_{k=1}\boldm^{(e)}_k\1_{10}(\boldd)_k \bmod n_e,
\end{align}
where~$a_e$ is the target value to be found in matrix~$A^{(e)}$. Eq.~\eqref{equation:parityincrease} implies that~$\boldd(i,j)$ satisfies the~$f$ redundancy.
\begin{center}
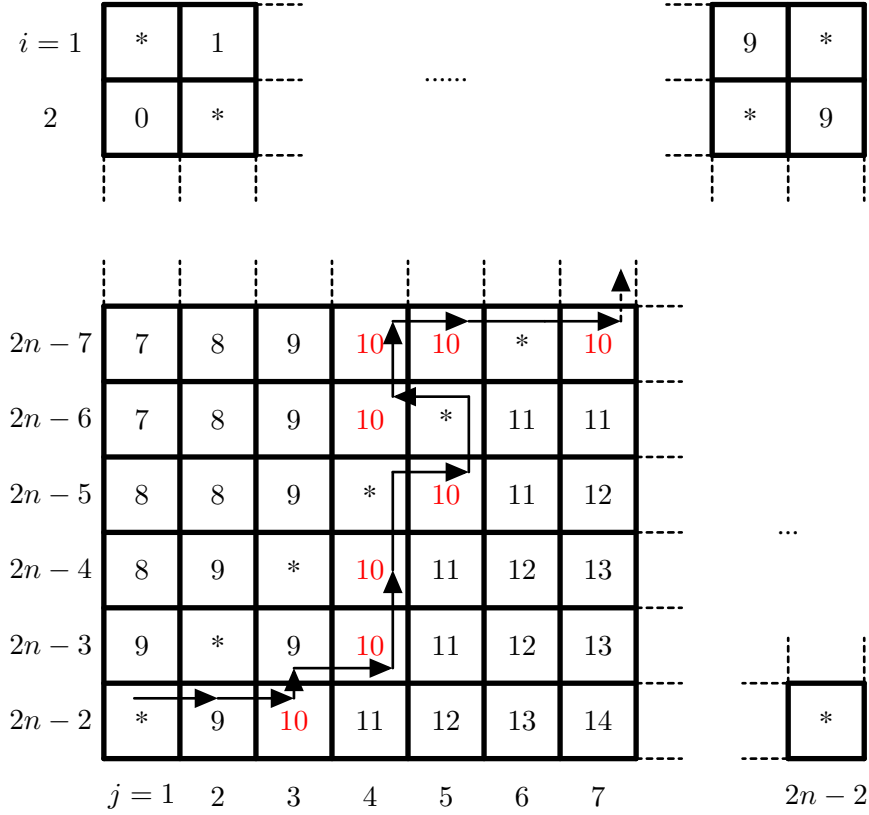
\begin{figure}
\definecolor{ffqqqq}{rgb}{1,0,0}
\begin{tikzpicture}[line cap=round,line join=round,>=triangle 45,x=2cm,y=2cm]
% \begin{axis}[
% x=1cm,y=1cm,
% axis lines=middle,
% xmin=-4.269036196985372,
% xmax=4.839208878227783,
% ymin=-3.3574428116815502,
% ymax=3.3672809873556426,
% xtick={-4,-3.5,...,4.5},
% ytick={-3,-2.5,...,3},]
\clip(-5,-3) rectangle (2.6,2.6);
\draw [line width=2pt] (-3,-1.5)-- (0.5,-1.5);
\draw [line width=2pt] (-3,-2)-- (0.5,-2);
\draw [line width=2pt] (-3,-1)-- (0.5,-1);
\draw [line width=2pt] (-3,-0.5)-- (0.5,-0.5);
\draw [line width=2pt] (-3,0)-- (0.5,0);
\draw [line width=2pt] (-2.5,-2.5)-- (-2.5,0.5);
\draw [line width=2pt] (-2,-2.5)-- (-2,0.5);
\draw [line width=2pt] (-1.5,0.5)-- (-1.5,-2.5);
\draw [line width=2pt] (-1,0.5)-- (-1,-2.5);
\draw [line width=2pt] (-0.5,0.5)-- (-0.5,-2.5);
\draw [line width=2pt] (0,0.5)-- (0,-2.5);
\draw [line width=2pt] (0.5,-2.5)-- (0.5,0.5);
\draw (-2.75,-2.25) node[anchor=center] {*};
\draw (-2.25,-1.75) node[anchor=center] {*};
\draw (-1.75,-1.25) node[anchor=center] {*};
\draw (-1.25,-0.75) node[anchor=center] {*};
\draw [line width=2pt] (-3,1.5)-- (-2,1.5);
\draw (-0.75,-0.25) node[anchor=center] {*};
\draw [line width=2pt] (1.5,-2)-- (1.5,-2.5);
\draw (1.75,-2.25) node[anchor=center] {*};
\draw (1.75,2.25) node[anchor=center] {*};
\draw (1.25,1.75) node[anchor=center] {*};
\draw (-2.75,2.25) node[anchor=center] {*};
\draw (-2.25,1.75) node[anchor=center] {*};
\draw [color=ffqqqq](0.25,0.25) node[anchor=center] {$10$};
\draw [color=ffqqqq](-0.75,0.25) node[anchor=center] {$10$};
\draw [color=ffqqqq](-0.75,-0.75) node[anchor=center] {$10$};
\draw [color=ffqqqq](-1.25,-1.75) node[anchor=center] {$10$};
\draw [color=ffqqqq](-1.75,-2.25) node[anchor=center] {$10$};
\draw (-0.25,0.25) node[anchor=center] {*};
\draw [color=ffqqqq](-1.25,0.25) node[anchor=center] {$10$};
\draw [color=ffqqqq](-1.25,-0.25) node[anchor=center] {$10$};
\draw [color=ffqqqq](-1.25,-1.25) node[anchor=center] {$10$};
\draw (-0.25,-0.25) node[anchor=center] {$11$};
\draw (0.25,-0.25) node[anchor=center] {$11$};
\draw (-0.25,-0.75) node[anchor=center] {$11$};
\draw (0.25,-0.75) node[anchor=center] {$12$};
\draw (-0.75,-1.25) node[anchor=center] {$11$};
\draw (-0.25,-1.25) node[anchor=center] {$12$};
\draw (0.25,-1.25) node[anchor=center] {$13$};
\draw (0.25,-1.75) node[anchor=center] {$13$};
\draw (-0.75,-1.75) node[anchor=center] {$11$};
\draw (-0.25,-1.75) node[anchor=center] {$12$};
\draw (-0.75,-2.25) node[anchor=center] {$12$};
\draw (-0.25,-2.25) node[anchor=center] {$13$};
\draw (0.25,-2.25) node[anchor=center] {$14$};
\draw (-1.25,-2.25) node[anchor=center] {$11$};
\draw (-1.75,0.25) node[anchor=center] {$9$};
\draw (-1.75,-0.25) node[anchor=center] {$9$};
\draw (-1.75,-0.75) node[anchor=center] {$9$};
\draw (-1.75,-1.75) node[anchor=center] {$9$};
\draw (-2.25,0.25) node[anchor=center] {$8$};
\draw (-2.25,-0.25) node[anchor=center] {$8$};
\draw (-2.25,-0.75) node[anchor=center] {$8$};
\draw (-2.75,-1.75) node[anchor=center] {$9$};
\draw (-2.25,-2.25) node[anchor=center] {$9$};
\draw (-2.75,0.25) node[anchor=center] {$7$};
\draw (-2.75,-0.25) node[anchor=center] {$7$};
\draw (-2.75,-0.75) node[anchor=center] {$8$};
\draw (-2.75,-1.25) node[anchor=center] {$8$};
\draw (-2.25,-1.25) node[anchor=center] {$9$};
\draw (-2.75,1.75) node[anchor=center] {$0$};
\draw (-2.25,2.25) node[anchor=center] {$1$};
\draw (1.25,2.25) node[anchor=center] {$9$};
\draw (1.75,1.75) node[anchor=center] {$9$};
\draw [line width=2pt] (-3,2.5)-- (-2,2.5);
\draw [line width=2pt] (-2,2.5)-- (-2,1.5);
\draw [line width=2pt] (-3,2.5)-- (-3,1.5);
\draw [line width=2pt] (-2.5,2.5)-- (-2.5,1.5);
\draw [line width=2pt] (-3,2)-- (-2,2);
\draw [line width=2pt] (1,2.5)-- (2,2.5);
\draw [line width=2pt] (2,2.5)-- (2,1.5);
\draw [line width=2pt] (1,1.5)-- (2,1.5);
\draw [line width=2pt] (1,2.5)-- (1,1.5);
\draw [line width=2pt] (1,2)-- (2,2);
\draw [line width=2pt] (1.5,2.5)-- (1.5,1.5);
\draw [line width=2pt] (1.5,-2)-- (2,-2);
\draw [line width=2pt] (1.5,-2.5)-- (2,-2.5);
\draw [line width=2pt] (2,-2.5)-- (2,-2);
\draw [line width=2pt] (-3,0.5)-- (-3,-2.5);
\draw [line width=2pt] (-3,-2.5)-- (0.5,-2.5);
\draw [line width=2pt] (-3,0.5)-- (0.5,0.5);
\draw (-3.35,-2.25) node[anchor=center] {$2n-2$};
\draw (-3.35,-1.75) node[anchor=center] {$2n-3$};
\draw (-3.35,-1.25) node[anchor=center] {$2n-4$};
\draw (-3.35,-0.75) node[anchor=center] {$2n-5$};
\draw (-3.35,-0.25) node[anchor=center] {$2n-6$};
\draw (-3.35,0.25) node[anchor=center] {$2n-7$};
\draw (-3.35,1.75) node[anchor=center] {$2$};
\draw (-3.35,2.25) node[anchor=center] {$i=1$};
\draw (-2.75,-2.75) node[anchor=center] {$j=1$};
\draw (-2.25,-2.75) node[anchor=center] {$2$};
\draw (-1.75,-2.75) node[anchor=center] {$3$};
\draw (-1.25,-2.75) node[anchor=center] {$4$};
\draw (-0.75,-2.75) node[anchor=center] {$5$};
\draw (-0.25,-2.75) node[anchor=center] {$6$};
\draw (0.25,-2.75) node[anchor=center] {$7$};
\draw (1.75,-2.75) node[anchor=center] {$2n-2$};
\draw [line width=1pt,dash pattern=on 2pt off 2pt] (-3,1.5)-- (-3,1.2);
\draw [line width=1pt,dash pattern=on 2pt off 2pt] (-2.5,1.5)-- (-2.5,1.2);
\draw [line width=1pt,dash pattern=on 2pt off 2pt] (-2,1.5)-- (-2,1.2);
\draw [line width=1pt,dash pattern=on 2pt off 2pt] (-3,0.5)-- (-3,0.8);
\draw [line width=1pt,dash pattern=on 2pt off 2pt] (-2.5,0.5)-- (-2.5,0.8);
\draw [line width=1pt,dash pattern=on 2pt off 2pt] (-2,0.5)-- (-2,0.8);
\draw [line width=1pt,dash pattern=on 2pt off 2pt] (-1.5,0.5)-- (-1.5,0.8);
\draw [line width=1pt,dash pattern=on 2pt off 2pt] (-1,0.5)-- (-1,0.8);
\draw [line width=1pt,dash pattern=on 2pt off 2pt] (-2,2.5)-- (-1.7,2.5);
\draw [line width=1pt,dash pattern=on 2pt off 2pt] (-2,2)-- (-1.7,2);
\draw [line width=1pt,dash pattern=on 2pt off 2pt] (-2,1.5)-- (-1.7,1.5);
\draw [line width=1pt,dash pattern=on 2pt off 2pt] (1,2.5)-- (0.7,2.5);
\draw [line width=1pt,dash pattern=on 2pt off 2pt] (1,2)-- (0.7,2);
\draw [line width=1pt,dash pattern=on 2pt off 2pt] (1,1.5)-- (0.7,1.5);
\draw [line width=1pt,dash pattern=on 2pt off 2pt] (1,1.5)-- (1,1.2);
\draw [line width=1pt,dash pattern=on 2pt off 2pt] (1.5,1.5)-- (1.5,1.2);
\draw [line width=1pt,dash pattern=on 2pt off 2pt] (2,1.5)-- (2,1.2);
\draw [line width=1pt,dash pattern=on 2pt off 2pt] (-0.5,0.5)-- (-0.5,0.8);
\draw [line width=1pt,dash pattern=on 2pt off 2pt] (0,0.5)-- (0,0.8);
\draw [line width=1pt,dash pattern=on 2pt off 2pt] (0.5,0.5)-- (0.5,0.8);
\draw [line width=1pt,dash pattern=on 2pt off 2pt] (0.5,0.5)-- (0.8,0.5);
\draw [line width=1pt,dash pattern=on 2pt off 2pt] (0.5,0)-- (0.8,0);
\draw [line width=1pt,dash pattern=on 2pt off 2pt] (0.5,-0.5)-- (0.8,-0.5);
\draw [line width=1pt,dash pattern=on 2pt off 2pt] (0.5,-1)-- (0.8,-1);
\draw [line width=1pt,dash pattern=on 2pt off 2pt] (0.5,-1.5)-- (0.8,-1.5);
\draw [line width=1pt,dash pattern=on 2pt off 2pt] (0.5,-2)-- (0.8,-2);
\draw [line width=1pt,dash pattern=on 2pt off 2pt] (0.5,-2.5)-- (0.8,-2.5);
\draw [line width=1pt,dash pattern=on 2pt off 2pt] (1.5,-2.5)-- (1.2,-2.5);
\draw [line width=1pt,dash pattern=on 2pt off 2pt] (1.5,-2)-- (1.5,-1.7);
\draw [line width=1pt,dash pattern=on 2pt off 2pt] (2,-2)-- (2,-1.7);
\draw [line width=1pt,dash pattern=on 2pt off 2pt] (1.5,-2)-- (1.2,-2);
\draw (-0.75,2) node[anchor=center] {......};
\draw (1.5,-1) node[anchor=center] {...};
\draw [->,line width=1pt] (-2.8,-2.1) -- (-2.25,-2.1);
\draw [->,line width=1pt] (-2.25,-2.1) -- (-1.75,-2.1);
\draw [->,line width=1pt] (-1.75,-2.1) -- (-1.75,-1.9);
\draw [->,line width=1pt] (-1.75,-1.9) -- (-1.1,-1.9);
\draw [->,line width=1pt] (-1.1,-1.9) -- (-1.1,-1.25);
\draw [line width=1pt] (-1.1,-1.25) -- (-1.1,-0.6);
\draw [->,line width=1pt] (-1.1,-0.6) -- (-0.6,-0.6);
\draw [line width=1pt] (-0.6,-0.6) -- (-0.6,-0.1);
\draw [->,line width=1pt] (-0.6,-0.1) -- (-1.1,-0.1);
\draw [->,line width=1pt] (-1.1,-0.1) -- (-1.1,0.4);
\draw [->,line width=1pt] (-1.1,0.4) -- (-0.6,0.4);
\draw [line width=1pt] (-0.6,0.4) -- (-0.1,0.4);
\draw [->,line width=1pt] (-0.1,0.4) -- (0.4,0.4);
\draw [->,line width=1pt,dash pattern=on 2pt off 2pt] (0.4,0.4) -- (0.4,0.75);
%\end{axis}
\end{tikzpicture}\caption{The path of Algorithm~\ref{algorithm:find10} on the matrix~$A^{(0)}$. The algorithm searches for all~$i,j$ pairs such that $A^{(0)}_{i,j}=10$ that appear in the lowest position (with maximum $i$) of each column. The algorithm proceeds right until the next term $A^{(0)}_{i,j}$ is greater than~$10$. Then, it proceeds up one step and repeats the process in the same manner.%\NetanelComment{Well done! Just please center it.}
}
\label{figure:matrix}
\end{figure}
\end{center}
The procedure to find such~$\boldd(i,j)$ is given in Algorithm~\ref{algorithm:find10}.
We search for all sequences~$\boldd(i,j)\ne \star$ with no adjacent $1$'s (satisfies \eqref{equation:checking11}) such that~$A^{(0)}_{i,j}=a_0$. This clearly amounts to a binary~search in a sorted matrix\footnote{The two~$\star$ entries in each row or column can simply be ignored.}.
%\NetanelRemove{The difference is that in this matrix~$A^0_{i,j}$, some of the entries do not have values. Fortunately, we have two~$*$s in each row.}~\NetanelRemove{We start from~$A^0_{1,2n}$, a corner of the matrix. Each step we either go one step right (
%increase the row index~$i$ by~$1$) or until~$A^0_{i,j}>a_0$, we go back to the leftmost~$A^0_{i,j}$ such that~$A^0_{i,j}\le a_0$ and then go one step up (
%decrease the column index~$i$~by$1$).
%or decrease the column index~$j$ by~$1$.}
We start from the bottom left corner of the matrix, proceed to the right in each step until reaching the rightmost entry such that~$A^{(0)}_{i,j}\le a_0$, and then go one step up.%\NetanelComment{This doesn't seem to match the figure. We are going more than one step up in several places, and it is not clear why we go right in entry~$(2n-3,4)$.}
~Figure~\ref{figure:matrix} illustrates an example of how Algorithm~\ref{algorithm:find10} runs on matrix~$A^{(0)}$.

To avoid the computation of the entire matrix, that would require~$O(n^2)$ time, each entry is computed from previously seen ones \textit{only} upon its discovery. To this end we prove the following lemma, that alongside Proposition~\ref{proposition:matrixAF}, provides a way of computing a newly discovered entry. 
\begin{lemma}
Whenever the~$(i,j)$-th and~$(i+1,j)$-th (resp.~$(i,j+1)$) entries of~$A^{(e)}$ are not~$\star$, we have that
\begin{alignat}{2}
A^{(e)}_{i,j}-A^{(e)}_{i+1,j}&=&\;&
b_i\boldm^{(e)}_{p_i}-b_{i+1}\boldm^{(e)}_{p_{i+1}}\nonumber\\
&~&& +\sum^{\min\{p_1,p_{i+1}\}}_{k=\min\{p_i,p_{i+1}\}-1}\1_{10}(\boldd)_k[(k+1)^e(\delta(\min\{p_i,p_j\}<k+1)-\delta(\min\{p_{i+1},p_j\}<k+1))\nonumber\\
&~&&+(k+2)^e(\delta(\max\{p_i,p_j\}<k+2)-\delta(\max\{p_{i+1},p_j\}<k+2))]\label{equation:columnupdate},\mbox{ and}\\
A^{(e)}_{i,j}-A^{(e)}_{i,j+1}&=&&
b_j\boldm^{(e)}_{p_j}-b_{j+1}\boldm^{(e)}_{p_{j+1}}\nonumber\\
&~&&+\sum^{\min\{p_j,p_{j+1}\}}_{k=\min\{p_j,p_{j+1}\}-1}\1_{10}(\boldd)_k[(k+1)^e(\delta(\min\{p_i,p_j\}<k+1)-\delta(\min\{p_i,p_{j+1}\}<k+1))\nonumber\\
&~&&+(k+2)^e(\delta(\max\{p_i,p_j\}<k+2)-\delta(\max\{p_i,p_{j+1}\}<k+2))]\label{equation:rowupdate}
\end{alignat}
\end{lemma}
\begin{proof}
Note that~if~$i$ increases by~$1$ or if~$j$ decreases by~$1$, then~$p_i$ or~$p_j$~changes by at most~$1$ (See \eqref{equation:pipj}). Hence,
\begin{align*}
&\delta(\min\{p_i,p_j\}<k+1)=\delta(\min\{p_{i+1},p_j\}<k+1),\\
&\delta(\max\{p_i,p_j\}<k+2)=\delta(\max\{p_{i+1},p_j\}<k+2)
\end{align*}
for~$k\le \min\{p_j,p_{i+1}\}-2$ and~$k\ge \min\{p_i,p_{i+1}\}+1$.
According to~\eqref{equation:increasematrix}, we have that~\eqref{equation:columnupdate} holds, and similarly,~\eqref{equation:rowupdate} holds as well.
\end{proof}
%\subsection{}
%Now we are ready to present the algorithm. %we use~$flag$ to denote whether~$F^1_{i,j}>a_0$.
%While~$F^0_{i,j}\le a_0$, we increase~$i$ by~$1$ until~$F^1_{i,j}>a_0$.
%Then we decrease~$i$ by $1$ and then decrease~$j$ by~$1$. We use~$x_e$ to update~$F^e_{i,j}$.
\begin{algorithm}[h]\label{algorithm:find10}

\KwIn{Subsequence~$\boldd\in\{0,1\}^{n-2}$ of~$\boldc$, and~$f(\boldc)$}
\KwOut{$i$ and~$j$ such that~$\boldd(i,j)=\1_{10}(c)$}
\textbf{Initialization:} $i=2n-2,j=1$\;
$x_e=A^{(e)}_{1,2n-2}$~for~$e\in\{0,1,2\}$\;
$a_e=f_e(\boldc)-\sum^{n-3}_{k=1}\boldm^{(e)}_k\1_{10}(\boldd)_k \bmod n_e$ for~$e\in\{0,1,2\}$\;
%\While{$x_0>a_0$ }{
% Compute $x_e= x_e+ F^e_{i,j-1}-F^e_{i,j}$, $e=0,1,2$\;
% $j=j-1$\;
%  \eIf{p(O(j-1))= p(O(i))}{
%      $x_e= x_e+ F^e_{i,j-1}-F^e_{i,j}$, $e=0,1,2$\;
%       $j=j-1$\;
%      }{}
% , $e=0,1,2$\;
% }
% \eIf{$x_e<a_e$, $e=0,1,2$}{
%      return $i,j$\;
%      }{}
\While{ $i\ge 0$}{
	\eIf{$x_e=a_e$ for every~$e\in\{0,1,2\}$ and~$\boldd(i,j)\ne \star$ and has no adjacent~$1$s' (satisfies \eqref{equation:checking11})}
      {
      return $i,j$\;
     }{
Find~the maximum $j$ for which~$A^{(0)}_{i,j}\le a_0$.

\eIf{$p_i=p_j$ or~($x_0>a_0$)}{
$temp\_x_e= x_e+ A^{(e)}_{i,j-1}-A^{(e)}_{i,j}$ (using~\eqref{equation:rowupdate}), for~$e\in\{0,1,2\}$\;
    $temp\_j=j-1$\;
     \While{ $p_{temp\_j}=p_i$}
    {$temp\_x_e= x_e+ A^{(e)}_{i,temp\_j-1}-A^{(e)}_{i,temp\_j}$ (using~\eqref{equation:rowupdate}) for~$e\in\{0,1,2\}$\;
    $temp\_j=temp\_j-1$\;}
    \If{$temp\_j\ge 1$}{
    $j=temp\_j$\;
    $x_e=temp\_x_e$~for~$e\in\{0,1,2,\}$\;
    }
    }{

$temp\_x_e= x_e+ A^{(e)}_{i,j+1}-A^{(e)}_{i,j}$ (using~\eqref{equation:rowupdate}), for~$e\in\{0,1,2\}$\;
    $temp\_j=j+1$\;
     \While{ $p_{temp\_j}=p_i$}
    {$temp\_x_e= x_e+ A^{(e)}_{i,temp\_j+1}-A^{(e)}_{i,temp\_j}$ (using~\eqref{equation:rowupdate}) for~$e\in\{0,1,2\}$\;
    $temp\_j=temp\_j+1$;}

\eIf{$temp\_x_0\le a_0$}{
   $j=temp\_j$\;
   $x_e=temp\_x_e$~for~$e\in\{0,1,2,\}$\;

   }{
   $x_e= x_e+ A^{(e)}_{i-1,j}-A^{(e)}_{i,j}$ (using~\eqref{equation:columnupdate})\;
   $i=i-1$\;
  }
    }}

 }
return $(0,0)$\;
    \caption{{\bf Finding~$\1_{10}(\boldc)$.} \label{Algorithm}}

\end{algorithm}

We first show that Algorithm~\ref{algorithm:find10} outputs the~$(i,j)$ pair such that~$\boldd(i,j)=\1_{10}(\boldc)$. Note that by Lemma~\ref{lemma:10indicator} there exists a unique sequence~$\boldd(i,j)=\1_{10}(\boldc)$ for which~$\boldd(i,j)$ satisfies Eq. \eqref{equation:checking11} and for which $(i,j)$ satisfies Eq. \eqref{equation:parityincrease}. Since the algorithm terminates either when such a sequence $\boldd(i,j)=\1_{10}(\boldc)$ is found or no such sequence is found and~$i$ reaches~$0$, it suffices to show that the latter case does not occur.~We prove this by contradiction. Assuming that the latter case occurs, we show that $\boldd(i,j)\ne \1_{10}(\boldc)$ for all~$(i,j)$ pairs, which is a contradiction. 
%that are not visited when the algorithm terminates belong to the following:
%\begin{enumerate}
%\item $A^{(e)}_{i,j}\ne a_e,$ for some~$e=0,1,2$
%\item $\boldd(i,j)=\star$ or~$\boldd(i,j)$ does not satisfy~\eqref{equation:checking11}
%\item $\boldd(i,j)=\boldd(i',j')$ for some visited~$(i',j')$ pair.
%\end{enumerate}
For each~$i\in\{1,2,\ldots,2n-2\}$, let $j_i$ be the maximum~$j=j_i$ for which $A^{(0)}_{i,j_i}\le a_0$. If~$A^{(0)}_{i,j}>a_0$ for all~$j$, then~$j_i=1$. Note that each pair~$(i,j_i)$ is visited in Algorithm~\ref{algorithm:find10} and by assumption we have that~$\boldd(i,j_i)\ne\1_{10}(\boldc)$.
We consider the following two cases
\begin{enumerate}
\item [$(1)$.] $j>j_i$
\item [$(2)$.] $j<j_i$ 
\end{enumerate}
and conclude that no~$(i,j)$ pairs in these cases result in~$\boldd(i,j)= \1_{10}(\boldc)$.
For $j>j_i$, by Proposition~\ref{proposition:nondecreasing} we have that~$A^{(0)}_{i,j}\ge A^{(0)}_{i,j_i}$ or that~$\boldd(i,j)=\star$. Hence by definition of~$j_i$
we have that~$A^{(0)}_{i,j}>a_0$ or that $\boldd(i,j)=\star$ and hence~$\boldd(i,j)\ne \1_{10}(\boldc)$. 
For~$j<j_i$, by Proposition~\ref{proposition:nondecreasing} we have that~$A^{(0)}_{i,j}\le A^{(0)}_{i,j_i}$ or that~$\boldd(i,j)=\star$. 
If~$A^{(0)}_{i,j}<A^{(0)}_{i,j_i}$, then~$A^{(0)}_{i,j}\ne a_0$. If~$A^{(0)}_{i,j}=A^{(0)}_{i,j_i}$, then according to Proposition~\ref{proposition:nondecreasing}, we have that~$\boldd(i,j)=\boldd(i,j_i)\ne \1_{10}(\boldc)$. 

We now show that Algorithm~\ref{algorithm:find10} terminates in~$O(n)$ time. 
From~\eqref{equation:columnupdate} and~\eqref{equation:rowupdate} the~$(i,j)$-th entry of~$A^{(e)},e\in\{0,1,2\}$, can be computed by using the update rule $x_e+ A^{(e)}_{i-1,j}-A^{(e)}_{i,j}$ and~$x_e+ A^{(e)}_{i,j\pm 1}-A^{(e)}_{i,j}$ (see Algorithm~\ref{algorithm:find10}), that can be computed in constant time. In addition, one can verify in constant time that~\eqref{equation:checking11} holds.

 Note that in each round, either~$i$ decreases by~$1$ or~$j$ increases by~$1$, with the exception that~$j$ decreases when~$A^{(0)}_{i,j}=\star$ or~$A^{(0)}_{i,j}>a_0$. We prove by contradiction that the latter case, in which~$A^{(0)}_{i,j}>a_0$ and~$j>1$ is impossible. Notice that for each current pair~$(i,j)$, the value of next pair~$(i^*,j^*)$ falls into either one of the following three case:
\begin{enumerate}
\item [$(1)$.] $(i^*,j^*)=(i,j')$ for some~$j'>j$ with~$A^{(0)}_{i^*,j^*}\le a_0$
\item [$(2)$.] $(i^*,j^*)=(i-1,j)$
\item [$(3)$.] $(i^*,j^*)=(i-1,j')$ for some~$j'<j$ when~$A^{(0)}_{i-1,j}=\star$.
\end{enumerate}

Assume by contradiction that~$A_{i^*,j^*}^{(0)}>a_0$ and~$j^*>1$, and $(i^*,j^*)$ is the first pair for which this statement is true. In Case~$(1)$, we have that~$A^{(0)}_{i^*,j^*}\le a_0$, in contradiction to~$A^{(0)}_{i^*,j^*}>a_0$. In Case~$(2)$ or Case~$(3)$, Proposition~\ref{proposition:nondecreasing} implies that~$a_0<A^{(0)}_{i^*,j^*}\le A^{(0)}_{i,j}$, contradicting the assumption that~$(i^*,j^*)$ is the first visited pair which satisfies $A^{(0)}_{i^*,j^*}>a_0$.

Having proved that~$A_{i,j}^{(0)}\le a_0$ whenever~$j>1$, we have the Algorithm~\ref{algorithm:find10} proceeds to left only when it encounters a~$\star$-entry. We now show that the algorithm terminates in~$O(n)$ time. Notice that unless Algorithm~\ref{algorithm:find10} encounters a~$\star$-entry, it proceeds either up or to the right, for which case, it is clear that only~$O(n)$ many steps occur. In cases where Algorithm~\ref{algorithm:find10} encounters a~$\star$-entry, it proceeds to the \textit{left} until a non~$\star$-entry is found. Since the number of~$\star$-entries is~$4n-4$,
the number of left strides of the algorithm is at most this quantity, and therefore
the algorithm terminates in at most~$O(n)$ time. In the following, we provide a running example of Algorithm~\ref
{algorithm:find10}.
\begin{example}
Consider~a sequence~$\boldc=(1,1,0,0,1,0,1,0)$, where the first and the~$6$-th bits are deleted, resulting in~$\boldd=(1,0,0,1,1,0)$. Then~$n=8$,~$\1_{10}(\boldc)=(0,1,0,0,1,0,1)$,~$f(\boldc)=(14,46,200)$, and~$\1_{10}(\boldd)=(1,0,0,0,1)$. Hence~$a_0=8,a_1=30,a_2=144$.

Then, Algorithm~\ref{algorithm:find10} proceeds in the following manner.
\begin{align*}
&i=1,j=14 ,p_i=p_j,x_0=7,x_1=28,x_2=140\\
\rightarrow &
i=2,j=14 ,\boldd(i,j)=(1,0,0,0,1,\underline{0},\underline{1}),x_0=7,x_1=28,x_2=140  \\ \rightarrow &
i=3,j=14 ,\boldd(i,j)=(1,0,0,0,\underline{0},1,\underline{1}),x_0=8,x_1=34,x_2=176, \\ \rightarrow &
i=4,j=14 ,\boldd(i,j)=(1,0,0,\underline{0},0,1,\underline{1}),x_0=8,x_1=34,x_2=176, \\ \rightarrow &
i=5,j=14 ,\boldd(i,j)=(1,0,\underline{0},0,0,1,\underline{1}),x_0=8,x_1=34,x_2=176, \\ \rightarrow &
i=6,j=14 ,\boldd(i,j)=(1,\underline{0},0,0,0,1,\underline{1}),x_0=8,x_1=34,x_2=176, \\ \rightarrow &
i=7,j=14 ,\boldd(i,j)=(\underline{0},1,0,0,0,1,\underline{1}),x_0=9,x_1=36,x_2=180 \\ \rightarrow &
i=7,j=13 ,\boldd(i,j)=(\underline{0},1,0,0,0,\underline{1},1),x_0=9,x_1=36,x_2=180 \\ \rightarrow &
i=7,j=12 ,\boldd(i,j)=(\underline{0},1,0,0,\underline{1},0,1),x_0=8,x_1=30,x_2=144
\end{align*}
\end{example}
\subsection{Recover the~original sequence~$\boldc$}
Let~$(i,j)$ be the output of Algorithm~\ref{algorithm:find10}, for which~we have that~$\boldd(i,j)=\1_{10}(\boldc)$. Let~$\boldc'$ be a length~$n$ supersequence after two insertions to~$\boldd$ such that~$\1_{10}(\boldc')=\1_{10}(\boldc)$. If~$b_i=1$, then inserting~$b_i$ to~$\1_{10}(\boldd)$ corresponds to either inserting a~$0$ to~$\boldd$ as the~$p_i+1$-th bit in~$\boldc'$ or inserting a~$1$ to~$\boldd$ as the~$p_i$-th bit in~$\boldc'$ (see Table~\ref{table:indicatordeletion}). If~$b_i=0$, then inserting~$b_i$ to~$\1_{10}(\boldd)$ corresponds to inserting a~$0$ or~$1$ in the first~$0$ run or~$1$ run respectively after the~$k'$-th bit in~$\boldc'$, where~$k'=\max_{k}\{\boldd(i,j)_k=1,k<p_i\}$. The same arguments hold for the insertion of~$b_j$.

Therefore, given the~$(i,j)$ pair that Algorithm~\ref{algorithm:find10} returns, there are at most four possible~$\boldc'$ supersequences of~$\boldd$ such that~$\1_{10}(\boldc')=\1_{10}(\boldc)$. One can check if the~$\boldc'$ sequences satisfy~$h(\boldc)$. According to Theorem~\ref{theorem:main}, there is a unique such sequence, the original sequence~$\boldc$ that satisfies both~$f(\boldc)$ and~$h(\boldc)$ simultaneously.
% \begin{algorithm}[H]
% \SetAlgoLined
% \KwResult{Write here the result }
%  initialization\;
%  \While{While condition}{
%   instructions\;
%   \eIf{condition}{
%    instructions1\;
%    instructions2\;
%    }{
%    instructions3\;
%   }
%  }
%  \caption{How to write algorithms}
% \end{algorithm}
%The idea is that we only have to enumerate all possible supersequences of~$\1_{10}(\boldd)$ that satisfy parity check~$f_1(\boldc)$.

\bibliographystyle{IEEEtran}
%\bibliography{IEEEabrv,deletion}

\begin{thebibliography}{1}
\bibitem{levenshtein1966binary}
V.~I. Levenshtein, ``Binary codes capable of correcting deletions, insertions, and reversals,'' in \emph{Soviet physics doklady}, vol.~10, no.~8, 1966, pp.~707--710.

\bibitem{vt1965}
R.~R. Varshamov and G.~M. Tenengolts, ``Codes which correct single asymmetric errors,'' in \emph{Autom. Remote Control}, vol.~26, no.~2, 1965, pp.~286--290.

\bibitem{helberg2002multiple}
A.~S. Helberg and H.~C. Ferreira, ``On multiple insertion/deletion correcting codes,'' \emph{IEEE Trans. on Inf. Th.}, vol.~48, no.~1, pp.~305--308, 2002.

\bibitem{paluncic2012multiple}
F.~Paluncic, K.~A. Abdel-Ghaffar, H.~C. Ferreira, and W.~A. Clarke, ``A multiple insertion/deletion correcting code for run-length limited sequences,'' \emph{IEEE Trans. on Inf. Th.}, vol.~58, no.~3, pp.~1809--1824, 2012.

\bibitem{brakensiek2016efficient}
J.~Brakensiek, V.~Guruswami, and S.~Zbarsky, ``Efficient low-redundancy codes for correcting multiple deletions,'' in \emph{Proceedings of the ACM-SIAM Symposium on Discrete Algorithms} (SODA), pp.~1884--1892. 2016

\bibitem{Ryan}
R.~Gabrys and F.~Sala, ``Codes correcting two deletions.''  \emph{arXiv:1712.07222} [cs.IT], 2017.
\end{thebibliography}

\section*{Appendix}
\underline{Proof of~\eqref{equation:g_m-Case(a)} (Case~(a))}:
\begin{align*}
&\;\phantom{=}(\mathbbm{1}_{10}(\boldc)-\mathbbm{1}_{10}(\boldc'))\cdot \boldm^{(e)}\\
&= \sum_{t=\ell_1}^{\ell_2}\left( \mathbbm{1}_{10}(\boldc)_t-\mathbbm{1}_{10}(\boldc')_t \right)\cdot (\boldm^{(e)})_t+\sum_{t=k_2}^{k_1}\left( \mathbbm{1}_{10}(\boldc)_t-\mathbbm{1}_{10}(\boldc')_t \right)\cdot (\boldm^{(e)})_t\\
&= (\mathbbm{1}_{10}(\boldc)_{\ell_2}-\mathbbm{1}_{10}(\boldc')_{\ell_2})\cdot(\boldm^{(e)})_{\ell_2} + (\mathbbm{1}_{10}(\boldc)_{k_1}-\mathbbm{1}_{10}(\boldc')_{k_1})\cdot(\boldm^{(e)})_{k_1}+\\
~&\phantom{\equiv} \sum_{t=\ell_1}^{\ell_2-1}\left( \mathbbm{1}_{10}(\boldc)_t-\mathbbm{1}_{10}(\boldc)_{t+1} \right)\cdot (\boldm^{(e)})_t+\sum_{t=k_2}^{k_1-1}\left( \mathbbm{1}_{10}(\boldc')_{t+1}-\mathbbm{1}_{10}(\boldc')_{t} \right)\cdot (\boldm^{(e)})_t\\
&= (\mathbbm{1}_{10}(\boldc)_{\ell_2}-\mathbbm{1}_{10}(\boldc')_{\ell_2})\cdot(\boldm^{(e)})_{\ell_2} + (\mathbbm{1}_{10}(\boldc)_{k_1}-\mathbbm{1}_{10}(\boldc')_{k_1})\cdot(\boldm^{(e)})_{k_1}+\\
~&\phantom{\equiv} \sum_{t=\ell_1}^{\ell_2-1}\mathbbm{1}_{10}(\boldc)_t\cdot(\boldm^{(e)})_t -\sum_{\ell_1+1}^{\ell_2}\mathbbm{1}_{10}(\boldc)_{t}\cdot (\boldm^{(e)})_{t-1}\\
~&\phantom{\equiv}+\sum_{t=k_2+1}^{k_1}\mathbbm{1}_{10}(\boldc')_{t}\cdot (\boldm^{(e)})_{t-1}-\sum_{t=k_2}^{k_1-1}\mathbbm{1}_{10}(\boldc')_{t} \cdot (\boldm^{(e)})_t\\
&= (\mathbbm{1}_{10}(\boldc)_{\ell_2}-\mathbbm{1}_{10}(\boldc')_{\ell_2})\cdot(\boldm^{(e)})_{\ell_2} + (\mathbbm{1}_{10}(\boldc)_{k_1}-\mathbbm{1}_{10}(\boldc')_{k_1})\cdot(\boldm^{(e)})_{k_1}+\\
~&\phantom{\equiv} \mathbbm{1}_{10}(\boldc)_{\ell_1}\cdot(\boldm^{(e)})_{\ell_1}-\mathbbm{1}_{10}(\boldc)_{\ell_2}\cdot(\boldm^{(e)})_{\ell_2-1}+\sum_{t=\ell_1+1}^{\ell_2-1}\mathbbm{1}_{10}(\boldc)_t\cdot t^e+\\
~&\phantom{\equiv} \mathbbm{1}_{10}(\boldc')_{k_1}\cdot(\boldm^{(e)})_{k_1-1}-\mathbbm{1}_{10}(\boldc')_{k_2}\cdot(\boldm^{(e)})_{k_2}-\sum_{t=k_2+1}^{k_1-1}\mathbbm{1}_{10}(\boldc')_t\cdot t^e\\%%%
&= (-\mathbbm{1}_{10}(\boldc')_{\ell_2})\cdot(\boldm^{(e)})_{\ell_2} + (\mathbbm{1}_{10}(\boldc)_{k_1})\cdot(\boldm^{(e)})_{k_1}+\\
~&\phantom{\equiv} \mathbbm{1}_{10}(\boldc)_{\ell_1}\cdot(\boldm^{(e)})_{\ell_1}+\sum_{t=\ell_1+1}^{\ell_2}\mathbbm{1}_{10}(\boldc)_t\cdot t^e-\mathbbm{1}_{10}(\boldc')_{k_2}\cdot(\boldm^{(e)})_{k_2}-\sum_{t=k_2+1}^{k_1}\mathbbm{1}_{10}(\boldc')_t\cdot t^e\\
~&= \mathbbm{1}_{10}(\boldc)_{\ell_1}\cdot(\boldm^{(e)})_{\ell_1}+\mathbbm{1}_{10}(\boldc)_{k_1}\cdot(\boldm^{(e)})_{k_1}+\sum_{t=\ell_1+1}^{\ell_2}\mathbbm{1}_{10}(\boldc)_t\cdot t^e\\
~&\phantom{\equiv} -\sum_{t=k_2+1}^{k_1}\mathbbm{1}_{10}(\boldc')_t\cdot t^e-\left( \mathbbm{1}_{10}(\boldc')_{\ell_2}\cdot(\boldm^{(e)})_{\ell_2}+\mathbbm{1}_{10}(\boldc')_{k_2}\cdot(\boldm^{(e)})_{k_2} \right)\\
~&= g_{\boldm^{(e)},\ell_1}(\mathbbm{1}_{10}(\boldc)_{\ell_1},\ldots,\mathbbm{1}_{10}(\boldc)_{\ell_2},\mathbbm{1}_{10}(\boldc')_{\ell_2})- g_{\boldm^{(e)},k_2}(\mathbbm{1}_{10}(\boldc')_{k_2},\ldots,\mathbbm{1}_{10}(\boldc')_{k_1},\mathbbm{1}_{10}(\boldc)_{k_1})
\end{align*}

\underline{Proof of~\eqref{equation:g_m-Case(b)} (Case (b))}:
\begin{align*}
&\;\phantom{=}(\mathbbm{1}_{10}(\boldc)-\mathbbm{1}_{10}(\boldc'))\cdot \boldm^{(e)}\\
&= \sum_{t=\ell_1}^{\ell_2}\left( \mathbbm{1}_{10}(\boldc)_t-\mathbbm{1}_{10}(\boldc')_t \right)\cdot (\boldm^{(e)})_t+\sum_{t=k_1}^{k_2}\left( \mathbbm{1}_{10}(\boldc)_t-\mathbbm{1}_{10}(\boldc')_t \right)\cdot (\boldm^{(e)})_t\\
&= (\mathbbm{1}_{10}(\boldc)_{\ell_2}-\mathbbm{1}_{10}(\boldc')_{\ell_2})\cdot(\boldm^{(e)})_{\ell_2} + (\mathbbm{1}_{10}(\boldc)_{k_2}-\mathbbm{1}_{10}(\boldc')_{k_2})\cdot(\boldm^{(e)})_{k_2}+\\
~&\phantom{\equiv} \sum_{t=\ell_1}^{\ell_2-1}\left( \mathbbm{1}_{10}(\boldc)_t-\mathbbm{1}_{10}(\boldc)_{t+1} \right)\cdot (\boldm^{(e)})_t+\sum_{t=k_1}^{k_2-1}\left( \mathbbm{1}_{10}(\boldc)_t-\mathbbm{1}_{10}(\boldc)_{t+1} \right)\cdot (\boldm^{(e)})_t\\
&= (\mathbbm{1}_{10}(\boldc)_{\ell_2}-\mathbbm{1}_{10}(\boldc')_{\ell_2})\cdot(\boldm^{(e)})_{\ell_2} + (\mathbbm{1}_{10}(\boldc)_{k_2}-\mathbbm{1}_{10}(\boldc')_{k_2})\cdot(\boldm^{(e)})_{k_2}+\\
~&\phantom{\equiv} \sum_{t=\ell_1}^{\ell_2-1}\mathbbm{1}_{10}(\boldc)_t\cdot(\boldm^{(e)})_t -\sum_{\ell_1+1}^{\ell_2}\mathbbm{1}_{10}(\boldc)_{t}\cdot (\boldm^{(e)})_{t-1}+	 \sum_{t=k_1}^{k_2-1}\mathbbm{1}_{10}(\boldc)_{t}\cdot (\boldm^{(e)})_{t}\\
~&\phantom{\equiv}-\sum_{t=k_1+1}^{k_2}\mathbbm{1}_{10}(\boldc)_{t} \cdot (\boldm^{(e)})_{t-1}\\
&= (\mathbbm{1}_{10}(\boldc)_{\ell_2}-\mathbbm{1}_{10}(\boldc')_{\ell_2})\cdot(\boldm^{(e)})_{\ell_2} + (\mathbbm{1}_{10}(\boldc)_{k_2}-\mathbbm{1}_{10}(\boldc')_{k_2})\cdot(\boldm^{(e)})_{k_2}+\\
~&\phantom{\equiv} \mathbbm{1}_{10}(\boldc)_{\ell_1}\cdot(\boldm^{(e)})_{\ell_1}-\mathbbm{1}_{10}(\boldc)_{\ell_2}\cdot(\boldm^{(e)})_{\ell_2-1}+\sum_{t=\ell_1+1}^{\ell_2-1}\mathbbm{1}_{10}(\boldc)_t\cdot t^e+\\
~&\phantom{\equiv} \mathbbm{1}_{10}(\boldc)_{k_1}\cdot(\boldm^{(e)})_{k_1}-\mathbbm{1}_{10}(\boldc)_{k_2}\cdot(\boldm^{(e)})_{k_2-1}+\sum_{t=k_1+1}^{k_2-1}\mathbbm{1}_{10}(\boldc)_t\cdot t^e\\
&= (-\mathbbm{1}_{10}(\boldc')_{\ell_2})\cdot(\boldm^{(e)})_{\ell_2} + (-\mathbbm{1}_{10}(\boldc')_{k_2})\cdot(\boldm^{(e)})_{k_2}+\\
~&\phantom{\equiv} \mathbbm{1}_{10}(\boldc)_{\ell_1}\cdot(\boldm^{(e)})_{\ell_1}+\sum_{t=\ell_1+1}^{\ell_2}\mathbbm{1}_{10}(\boldc)_t\cdot t^e+\mathbbm{1}_{10}(\boldc)_{k_1}\cdot(\boldm^{(e)})_{k_1}+\sum_{t=k_1+1}^{k_2}\mathbbm{1}_{10}(\boldc)_t\cdot t^e\\
~&= \mathbbm{1}_{10}(\boldc)_{\ell_1}\cdot(\boldm^{(e)})_{\ell_1}+\mathbbm{1}_{10}(\boldc)_{k_1}\cdot(\boldm^{(e)})_{k_1}-\left( \mathbbm{1}_{10}(\boldc')_{\ell_2}\cdot(\boldm^{(e)})_{\ell_2}+\mathbbm{1}_{10}(\boldc')_{k_2}\cdot(\boldm^{(e)})_{k_2} \right)+\\
~&\phantom{\equiv} \sum_{t=\ell_1+1}^{\ell_2}\mathbbm{1}_{10}(\boldc)_t\cdot t^e+\sum_{t=k_1+1}^{k_2}\mathbbm{1}_{10}(\boldc)_t\cdot t^e\\
~&= g_{\boldm^{(e)},\ell_1}(\mathbbm{1}_{10}(\boldc)_{\ell_1},\ldots,\mathbbm{1}_{10}(\boldc)_{\ell_2},\mathbbm{1}_{10}(\boldc')_{\ell_2})+ g_{\boldm^{(e)},k_1}(\mathbbm{1}_{10}(\boldc)_{k_1},\ldots,\mathbbm{1}_{10}(\boldc)_{k_2},\mathbbm{1}_{10}(\boldc')_{k_2})
\end{align*}

\underline{Proof of~\eqref{equation:g_m-Case(c)} (Case (c))}:
\begin{align*}
&\;\phantom{=}(\mathbbm{1}_{10}(\boldc)-\mathbbm{1}_{10}(\boldc'))\cdot \boldm^{(e)}\\
&=\sum_{t=\ell_1}^{k_1-2}(\mathbbm{1}_{10}(\boldc)_t-\mathbbm{1}_{10}(\boldc')_t)\cdot(\boldm^{(e)})_t+\sum_{t=k_1-1}^{\ell_2-1}(\mathbbm{1}_{10}(\boldc)_t-\mathbbm{1}_{10}(\boldc')_t)\cdot(\boldm^{(e)})_t\\
&\phantom{+}+\sum_{t=\ell_2}^{k_2}(\mathbbm{1}_{10}(\boldc)_t-\mathbbm{1}_{10}(\boldc')_t)\cdot(\boldm^{(e)})_t\\
&= \sum_{t=\ell_1}^{k_1-2}(\mathbbm{1}_{10}(\boldc)_t-\mathbbm{1}_{10}(\boldc)_{t+1})\cdot(\boldm^{(e)})_{t}+ \sum_{t=k_1-1}^{\ell_2-1}(\mathbbm{1}_{10}(\boldc)_t-\mathbbm{1}_{10}(\boldc)_{t+2})\cdot(\boldm^{(e)})_t+\\
&\phantom{+} (\mathbbm{1}_{10}(\boldc)_{\ell_2}-\mathbbm{1}_{10}(\boldc')_{\ell_2})\cdot(\boldm^{(e)})_{\ell_2}+(\mathbbm{1}_{10}(\boldc)_{k_2}-\mathbbm{1}_{10}(\boldc')_{k_2})\cdot(\boldm^{(e)})_{k_2}+\\
&\phantom{+} \sum_{t=\ell_2+1}^{k_2-1}(\mathbbm{1}_{10}(\boldc)_t-\mathbbm{1}_{10}(\boldc)_{t+1})\cdot(\boldm^{(e)})_t\\
&= \sum_{t=\ell_1}^{k_1-2}\mathbbm{1}_{10}(\boldc)_t\cdot(\boldm^{(e)})_{t}-\sum_{t=\ell_1+1}^{k_1-1}\mathbbm{1}_{10}(\boldc)_{t}\cdot(\boldm^{(e)})_{t-1}+\sum_{t=k_1-1}^{\ell_2}\mathbbm{1}_{10}(\boldc)_t\cdot(\boldm^{(e)})_t\\
&\phantom{+} -\sum_{t=k_1+1}^{\ell_2+1}\mathbbm{1}_{10}(\boldc)_{t}\cdot(\boldm^{(e)})_{t-2}+ (-\mathbbm{1}_{10}(\boldc')_{\ell_2})\cdot(\boldm^{(e)})_{\ell_2}+(-\mathbbm{1}_{10}(\boldc')_{k_2})\cdot(\boldm^{(e)})_{k_2}+\\
&\phantom{+} \sum_{t=\ell_2+1}^{k_2}\mathbbm{1}_{10}(\boldc)_t\cdot(\boldm^{(e)})_t-\sum_{t=\ell_2+2}^{k_2}\mathbbm{1}_{10}(\boldc)_{t}\cdot(\boldm^{(e)})_{t-1}\\
&= \mathbbm{1}_{10}(\boldc)_{\ell_1}(\boldm^{(e)})_{\ell_1}-\mathbbm{1}_{10}(\boldc)_{k_1-1}(\boldm^{(e)})_{k_1-2}+\sum_{t=\ell_1+1}^{k_1-2}\mathbbm{1}_{10}(\boldc)_t\cdot t^e+\\
&\phantom{+}\mathbbm{1}_{10}(\boldc)_{k_1-1}(\boldm^{(e)})_{k_1-1}+\mathbbm{1}_{10}(\boldc)_{k_1}(\boldm^{(e)})_{k_1}-\mathbbm{1}_{10}(\boldc)_{\ell_2+1}(\boldm^{(e)})_{\ell_2-1}\\
&\phantom{+}+\sum_{t=k_1+1}^{\ell_2}\mathbbm{1}_{10}(\boldc)_t(t^e+(t-1)^e)+ (-\mathbbm{1}_{10}(\boldc')_{\ell_2})\cdot(\boldm^{(e)})_{\ell_2}+(-\mathbbm{1}_{10}(\boldc')_{k_2})\cdot(\boldm^{(e)})_{k_2}+\\
&\phantom{+} \mathbbm{1}_{10}(\boldc)_{\ell_2+1}(\boldm^{(e)})_{\ell_2+1}+\sum_{t=\ell_2+2}^{k_2}\mathbbm{1}_{10}(\boldc)_t t^e\\%%%
&= \mathbbm{1}_{10}(\boldc)_{\ell_1}(\boldm^{(e)})_{\ell_1}+\mathbbm{1}_{10}(\boldc)_{k_1}(\boldm^{(e)})_{k_1}-(\mathbbm{1}_{10}(\boldc')_{\ell_2}\cdot(\boldm^{(e)})_{\ell_2}+\mathbbm{1}_{10}(\boldc')_{k_2}\cdot(\boldm^{(e)})_{k_2})+\\
&\phantom{+} \sum_{t=\ell_1+1}^{k_1-1}\mathbbm{1}_{10}(\boldc)_t\cdot t^e+\sum_{t=k_1+1}^{\ell_2+1}\mathbbm{1}_{10}(\boldc)_t(t^e+(t-1)^e)+\sum_{t=\ell_2+2}^{k_2}\mathbbm{1}_{10}(\boldc)_t t^e\\
&= \mathbbm{1}_{10}(\boldc)_{\ell_1}(\boldm^{(e)})_{\ell_1}+\mathbbm{1}_{10}(\boldc)_{k_1}(\boldm^{(e)})_{k_1}-(\mathbbm{1}_{10}(\boldc')_{\ell_2}\cdot(\boldm^{(e)})_{\ell_2}+\mathbbm{1}_{10}(\boldc')_{k_2}\cdot(\boldm^{(e)})_{k_2})+\\
&\phantom{+} \sum_{t=\ell_1+1}^{k_1-1}\mathbbm{1}_{10}(\boldc)_t\cdot t^e+\sum_{t=k_1}^{\ell_2}\mathbbm{1}_{10}(\boldc)_{t+1}t^e+\sum_{t=k_1+1}^{k_2}\mathbbm{1}_{10}(\boldc)_t t^e\\
&= g_{\boldm^{(e)},\ell_1}(\mathbbm{1}_{10}(\boldc)_{\ell_1},\ldots,\mathbbm{1}_{10}(\boldc)_{k_1-1},\mathbbm{1}_{10}(\boldc)_{k_1+1},\ldots,\mathbbm{1}_{10}(\boldc)_{\ell_2+1},\mathbbm{1}_{10}(\boldc')_{\ell_2}  )+\\
&\phantom{+}g_{\boldm^{(e)},k_1}( \mathbbm{1}_{10}(\boldc)_{k_1},\ldots,\mathbbm{1}_{10}(\boldc)_{k_2},\mathbbm{1}_{10}(\boldc')_{k_2})
\end{align*}
\end{document}